\documentclass[10pt,twocolumn,amsmath,amssymb,aps,pra,secnumarabic,
    nofootinbib,groupedaddress]{revtex4-2}

\usepackage{mathtools,mathptmx,amsthm,tikz,enumitem,eucal,stmaryrd}
\usepackage[tikz,rightline=false,leftline=false]{mdframed}
\usepackage[colorlinks=true,linkcolor=red!70!black,citecolor=blue!70!black,
  urlcolor=magenta!70!black,backref]{hyperref}
\setlist[enumerate,1]{font=\bfseries,label=\arabic*.}
\setlist[enumerate,2]{label=\arabic*.}
\setlist[enumerate,3]{font=\itshape,label=\arabic*.}
\numberwithin{equation}{section}

\makeatletter\def\@bibdataout@init{}\def\pre@bibdata{}\makeatother

\usetikzlibrary{decorations.pathreplacing}
\usetikzlibrary{arrows}
\colorlet{darkgreen}{green!50!black}
\colorlet{darkblue}{blue!70!black}
\colorlet{darkred}{red!70!black}
\colorlet{lightblue}{white!85!blue}
\colorlet{lightred}{white!70!red}

\newtheorem{theorem}{Theorem}[section]

\newtheorem{corollary}[theorem]{Corollary}
\newtheorem{lemma}[theorem]{Lemma}
\newtheorem{proposition}[theorem]{Proposition}
\theoremstyle{definition}

\newtheorem{algothm}{Algorithm}
\newenvironment{algorithm}[1]{\begin{mdframed}[linewidth=1pt]%
    \renewcommand\thealgothm{\textsf{#1}}\algothm}%
    {\endalgothm\end{mdframed}\vspace{.5\baselineskip}}
\theoremstyle{remark}
\newtheorem*{remark}{Remark}
\newtheorem*{example}{Example}
\newtheorem*{examples}{Examples}

\newcommand{\Cor}[1]{Corollary~\ref{#1}}

\newcommand{\Fig}[1]{Figure~\ref{#1}}
\newcommand{\Lem}[1]{Lemma~\ref{#1}}
\newcommand{\Prop}[1]{Proposition~\ref{#1}}
\newcommand{\Sec}[1]{Section~\ref{#1}}
\newcommand{\Thm}[1]{Theorem~\ref{#1}}
\newcommand{\Alg}[1]{Algorithm~\ref{#1}}
\newcommand{\equ}[1]{equation~\eqref{#1}}
\newcommand{\Equ}[1]{Equation~\eqref{#1}}

\newcommand{\ie}{\emph{i.e.}}

\newcommand{\etc}{\emph{etc.}}

\newcommand{\GL}{\operatorname{GL}}
\newcommand{\M}{\operatorname{M}}
\newcommand{\PGL}{\operatorname{PGL}}
\newcommand{\PSL}{\operatorname{PSL}}
\newcommand{\PSO}{\operatorname{PSO}}
\newcommand{\PSU}{\operatorname{PSU}}
\newcommand{\PU}{\operatorname{PU}}
\newcommand{\SL}{\operatorname{SL}}
\newcommand{\SO}{\operatorname{SO}}
\newcommand{\SU}{\operatorname{SU}}
\newcommand{\Spin}{\operatorname{Spin}}
\newcommand{\Sp}{\operatorname{Sp}}
\newcommand{\U}{\operatorname{U}}
\newcommand{\so}{\operatorname{so}}
\newcommand{\su}{\operatorname{su}}
\renewcommand{\O}{\operatorname{O}}
\renewcommand{\sl}{\operatorname{sl}}
\renewcommand{\u}{\operatorname{u}}

\newcommand{\ab}{\operatorname{ab}}
\newcommand{\amax}{\operatorname{amax}}
\newcommand{\bit}{\operatorname{bit}}
\newcommand{\can}{\operatorname{can}}
\newcommand{\ccan}{\operatorname{ccan}}
\newcommand{\len}{\operatorname{len}}
\newcommand{\nil}{\operatorname{nil}}
\newcommand{\poly}{\operatorname{poly}}
\newcommand{\tr}{\operatorname{tr}}
\renewcommand{\Re}{\operatorname{Re}}

\newcommand{\st}{\;\mid\;}

\newcommand{\eps}{\epsilon}
\newcommand{\into}{\hookrightarrow}
\newcommand{\longto}{\longrightarrow}
\newcommand{\onto}{\twoheadrightarrow}
\renewcommand{\setminus}{\smallsetminus}
\renewcommand{\tensor}{\otimes}

\newcommand{\bigcomm}[1]{\big\llbracket #1\big\rrbracket}
\newcommand{\braket}[1]{\langle #1\rangle}

\newcommand{\comm}[1]{\llbracket #1\rrbracket}

\providecommand{\abs}[1]{\lvert#1\rvert}
\providecommand{\norm}[1]{\lVert#1\rVert}

\newcommand{\C}{\mathbb{C}}

\newcommand{\Q}{\mathbb{Q}}
\newcommand{\R}{\mathbb{R}}
\newcommand{\Z}{\mathbb{Z}}

\newcommand{\hell}{\hat{\ell}}
\newcommand{\tG}{\widetilde{G}}
\newcommand{\tOmega}{\widetilde{\Omega}}
\newcommand{\tO}{\widetilde{O}}
\newcommand{\tTheta}{\widetilde{\Theta}}
\newcommand{\tomega}{\widetilde{\omega}}
\newcommand{\va}{\vec{a}}
\newcommand{\ve}{\vec{e}}
\newcommand{\vg}{\vec{g}}
\newcommand{\vt}{\vec{t}}
\newcommand{\vu}{\vec{u}}
\newcommand{\vv}{\vec{v}}
\newcommand{\vx}{\vec{x}}
\newcommand{\vy}{\vec{y}}
\renewcommand{\H}{\mathbb{H}}

\def\app#1#2{\mathrel{\setbox0=\hbox{$#1\sim$}%
    \setbox2=\hbox{\rlap{\hbox{$#1\propto$}}\lower1.1\ht0\box0}%
    \raise0.25\ht2\box2}}

\newcommand{\defeq}{\stackrel{\text{def}}=}

\newenvironment{ctp}{\[\begin{tikzpicture}}
    {\end{tikzpicture}\] \ignorespacesafterend}

\newenvironment{eq}[1]{\begin{equation} \label{#1}}
    {\end{equation}\ignorespacesafterend}
\newcommand{\eatline}{\vspace{-\baselineskip}}
\newcommand{\blankline}{\vspace{\baselineskip}}

\begin{document}
\title{Breaking the cubic barrier in the Solovay--Kitaev algorithm}
\author{Greg Kuperberg}
\email{greg@math.ucdavis.edu}
\thanks{Partly supported by NSF grants CCF-2009029 and CCF-2317280}
\affiliation{University of California, Davis}

\date{\today}

\begin{abstract}
\centerline{\textit{\normalsize Dedicated to the memory of
    Abdelrhman Elkasapy (1983--2017)}}
\blankline

We improve the Solovay--Kitaev theorem and algorithm for a general finite,
inverse-closed generating set acting on a qudit.  Prior versions of the
algorithm efficiently find a word of length $O(n^{3+\delta})$ to approximate
an arbitrary target gate to $n$ bits of precision.  Using two new ideas,
each of which reduces the exponent separately, our new bound on the word
length is $O(n^{1.44042\ldots+\delta})$.  Our result holds more generally
for any finite set that densely generates any connected, semisimple real
Lie group, with an extra length term in the noncompact case to reach
group elements far away from the identity.
\end{abstract}

\maketitle

\tableofcontents

\section{Introduction}
\label{s:intro}

The Solovay--Kitaev theorem \cite[Lem.~4.7]{Kitaev:error} is a foundational
result in quantum computing, and an important result in group theory.
In its algorithmic form, it says the following: Let $A$ be a finite set of
$d \times d$ unitary matrices that densely generate $\SU(d)$, and such that
$a \in A$ implies that $a^{-1} \in A$.  Then there is a polynomial-time
classical algorithm to approximate any $g \in \SU(d)$ with a product
\[ w = \prod_{j=1}^\ell a_j \qquad a_j \in A\]
of length $\ell = \poly(n)$, such that $w$ approximates $g$ to $n$ bits of
precision.  Taking $\norm{g-w} < \eps = 2^{-n}$ as the approximation error
(using some Banach norm on matrices), the length bound is also commonly
written as $\ell = \poly(\log(1\/\eps))$.

In the context of quantum computing, $A$ is interpreted as a gate set,
and the Solovay--Kitaev algorithm is needed to approximate an arbitrary
unitary operator $g \in \SU(d)$ (or more precisely $g \in \PU(d)$, which is
nearly equivalent) in a fault-tolerant manner.  Here, dense generation is
the only known way to circumvent the problem that if $g$ is selected from
a continuous family, then it cannot directly be fault tolerant, because
it cannot be protected from imprecision in its parameters.  The theorem
also provides an efficient gate-set independence result for the quantum
computing model, because it provides a way to approximate any gate set
$A$ with efficiently computable short words in any other gate set $B$.
In the context of group theory, the Solovay--Kitaev theorem is related
to other results that quantify how quickly a finite set $A \subset G$
densely generates a Lie group $G$; see \Sec{ss:approx}.

Different proofs of the Solovay--Kitaev theorem produce approximating words
of length $O(n^\alpha)$ for different values of $\alpha$.  (Note that
the asymptotic bounds are usually only uniform in $n$; the specific
constant and sometimes the exponent may depend on $d$ and $A$.)  To our
knowledge, the competitive word length bound in the general case in the
existing literature is $O(n^\alpha)$ for any $\alpha > 3$, as established
by Kitaev, Shen, and Vyalyi \cite[Thm.~8.5]{KSV:cqc}.  (See \Sec{ss:approx}
concerning special gate sets.)

Our main result is the following improvement to the algorithmic
Solovay--Kitaev theorem for general gate sets.

\begin{theorem} Let $A = A^{-1} \subset \SU(d)$ be a finite set of matrices
that densely generate $\SU(d)$ for some fixed $d$.  Let
\[ \alpha > \log_\phi(2) = 1.44042\ldots, \]
where $\phi = (1+\sqrt5)/2$ is the golden ratio, and let $\delta > 0$.
Then there is a polynomial time classical algorithm to approximate any $g
\in \SU(d)$ to $n$ bits of precision by a product $w$ in $A$ with word length
$\ell = O(n^\alpha)$ and compressed word length $\hell = O(n^{1+\delta})$.
\label{th:qudit} \end{theorem}

The statement of \Thm{th:qudit} mentions the concept of a compressed
word, by definition a word expressed as an algebraic circuit using the
concatenation operation; see \Sec{ss:models}.  As discussed there, $\hell =
\tOmega(n)$ for most $g \in \SU(d)$, so the bound on the compressed word
length is nearly optimal.  The compressed word length is also related to
the time complexity of the algorithm; see \Thm{th:runtime}.  (Dawson and
Nielsen \cite{DN:solkit} also use compressed words in their treatment of
the Solovay--Kitaev theorem.)

We also generalize \Thm{th:qudit} to an arbitrary connected semisimple
Lie group $G$ that may or may not be compact.

\begin{theorem} Let $G$ be a semisimple, connected real Lie group with a
left-invariant Finsler metric $d_G(\cdot,\cdot)$, and let $A = A^{-1} \subset
G$ densely generate $G$.  Let $\alpha > \log_\phi(2)$ and $\delta > 0$.
Then there is a polynomial-time algorithm to approximate any $g \in G$
by a product $w$ in $A$ with word length $\ell = O(n^\alpha + R)$
and compressed word length $\hell = O(n^{1+\delta})$,
where $d_G(g,w) < 2^{-n}$ and $d_G(g,1) < R$.
\label{th:semisimple} \end{theorem}

\begin{corollary} If $G \subseteq \SL(d,\C)$ is a semisimple, real matrix
Lie group in \Thm{th:semisimple}, then
\[ \ell = O((n + N)^\alpha), \]
where $\norm{g-w} < 2^{-n}$ and $\norm{g} < 2^N$.
\label{c:matnorm} \end{corollary}

We establish \Cor{c:matnorm} in \Sec{ss:metrics}.

We need to address several issues to interpret \Thm{th:semisimple} and
\Cor{c:matnorm}.
\begin{itemize}
\item We discuss left-invariant Finsler metrics in \Sec{ss:metrics}.
In particular, they are all equivalent for the purpose of Theorems
\ref{th:qudit} and \ref{th:semisimple}.
\item An element $g \in G$ requires a data representation when $G$ cannot
be realized as a group of matrices.  This is addressed in \Sec{ss:nonmat}.
\item The gates in $A$ and the target $g \in G$ have real parameters, which
also require a data representation.  This is addressed in \Sec{ss:models}.
\end{itemize}

Modulo these clarifications, we give a more precise description of the
input and time complexity of Theorems~\ref{th:qudit} and \ref{th:semisimple}.

\begin{theorem} For a general gate set $A$, the algorithm in
\Thm{th:semisimple} requires $n+O(\log(n)+R)$ bits of precision for each
gate $a \in A$ and $n+O(R)$ bits of precision for the target $g \in G$.
and then runs in time $\tO(n^{2+\delta}+(R+n)^2)$ for any $\delta > 0$.
If $G$ is an algebraic group and $A$ is an algebraic gate set, then the
algorithm runs in time $\tO(n^\alpha+R)$.
\label{th:runtime} \end{theorem}

\subsection{Ideas in the algorithm and proof}
\label{ss:ideas}

Our proofs of \Thm{th:qudit} and \Thm{th:semisimple} rest on two new ideas.

All proofs of the general Solovay--Kitaev theorem that we have seen
(including ours) are multiscale constructions.   We bring a target
group element $g \in G$ successively closer to the identity with a
recurrence of the form $g' = gs^{-1}$ for some step $s \approx g$.
In the Kitaev--Shen--Vyalyi version, a step $s$ with distance $O(2^{-n})$
has word length $O(n^{2+\delta})$.  A total word $w$ built from these
steps converges at a singly exponential rate, which means that a linear
number of steps are needed and $w$ has length $O(n^{3+\delta})$.  In the
Dawson--Nielsen approach \cite{DN:solkit}, the total word $w$ converges
doubly exponentially in the number of steps, so that each step $s$ and
the total word $w$ both have length $O(n^\alpha)$ with the same exponent
$\alpha$.  However, their doubly exponential convergence requires more
expensive steps, so that they only obtain
\[ \alpha = \log_{3/2}(5) = 3.96936\ldots. \]

One new idea in this article (\Sec{ss:zgolf}) is a multiscale framework that
combines the advantages of the Kitaev--Shen--Vyalyi and Dawson--Nielsen
algorithms.  We first formulate a recursive algorithm to create a step
$s$ (mainly using group commutators) with roughly exponential distance
$\Theta(2^{-n})$ and word length $O(n^2)$.  We then formulate a separate
recursive algorithm to precisely aim the steps (using group conjugation)
to achieve a total word $w$ that converges to $g$ at a doubly exponential
rate, so that $w$ also has length $O(n^2)$.  This approach appears to be
more economical than making precise steps solely from earlier precise steps,
as both Kitaev--Shen--Vyalyi and Dawson--Nielsen do.

If we make the roughly exponential steps using group commutators, then the
steps have word length $O(n^2)$ with an exponent of $2 = \log_2(4)$. More
explicitly, if two steps $u$ and $v$ have some length $\ell$ and distance
$\Theta(2^{-n})$, then their commutator $uvu^{-1}v^{-1}$ has length $4\ell$
and distance $\Theta(4^{-n})$.  Our other new idea (in \Sec{ss:steps}) is to
replace the ordinary group commutator by more efficient higher commutators
by borrowing from Elkasapy and Thom \cite{ET:lower,Elkasapy:lower}, who
studied the efficiency of higher group commutators in a different context.
They found a sequence of higher commutators whose efficiency converges
to $\log_\phi(2)$.  Using their higher commutators, we can achieve an
improved algorithm where the rough steps, precise steps, and total words
all have word length $O(n^\alpha)$ for any $\alpha > \log_\phi(2)$.

\acknowledgments

The author would like to especially thank Adam Bouland and Tudor
Giurgica-Tiron for discussions that led to the results in this article. The
author would also like to thank Indira Chatterji, Sami Douba, Jeff Danciger,
Sanchayan Dutta, Fran\c{c}ois Gueritaud, Fanny Kassel, Fran\c{c}ois Labourie,
and Daniel Tobias for useful discussions.  Finally, the author would like
to thank the anonymous referee for expert corrections to the manuscript.

The author learned from Andreas Thom that his former student Abdelrhman
Elkasapy died from an illness shortly after obtaining results which are
important for this article.  This article is dedicated to Elkasapy for
this reason.

\section{Context and history}
\label{s:history}

\subsection{Motivation for this work}
\label{ss:motive}

Bouland and Giurgica-Tiron \cite{BG:ifsk} recently found a generalization
of the Solovay--Kitaev theorem to the inverse-free case, \ie, in which $A$
densely generates $\SU(d)$ or $\SL(d,\C)$, without using the inverses of the
generators.  To prove their result, they find inverse-free self-cancelling
words that are more complicated than the group commutator $ghg^{-1}h^{-1}$.
Moreover, some inverse-free self-cancelling words are more efficient
than others.  This led the author to consider higher commutators for the
original Solovay--Kitaev problem with inverses, and then to discover the
important results of Elkasapy and Thom.  Later, while trying to simplify
the multiscale iteration, the author further improved the word length
exponent with doubly exponential convergence.

Without any control over the exponent $\alpha$, the author previously proved
\Thm{th:semisimple} for any perfect, connected, matrix Lie group $G$, with
a $\poly(n)$ length bound \cite[Thm.~2.4]{K:jones}.  (See also Aharonov,
Arad, Eban, and Landau \cite[Thm.~7.5]{AAEL:tutte}.)  However, the author
discovered mistakes in this previous work, which became another motivation
for the present work. The algorithm there constructs a frame at each scale
$\Theta(r^n)$ for some $r < 1$, and uses it to make a deformed lattice which
is an $O(r^n)$-net.  Since the algorithm makes precise steps from precise
steps, the lattices are part of the recursion to the steps, and the result
is a complicated exponent that depends on the specific Lie group $G$.
The theorem states and proves a valid length bound $\ell = \poly(n)$,
except that it is missing a macroscopic term when $G$ is not compact.
However, the last paragraph of the proof erroneously claims an $O(n^3)$
bound that does not hold for that algorithm.

Perfect Lie groups are more general and less robust than semisimple Lie
groups.  We do not know whether the Solovay--Kitaev exponent $\alpha$
is uniformly bounded for the class of perfect Lie groups.

\subsection{Modes of approximation}
\label{ss:approx}

Let $G$ be a connected Lie group, and let $A \subset G$ be a finite subset
such that the words in $A$ of all lengths are dense in $G$.  The set $A$
may or may not be closed under inverses; if it is, then we say that $A$
is \emph{symmetric}.  Regardless, we assume that $1 \in A$.  Let $A^\ell$
denote the words in $A$ of length at most $\ell$.  Given $g \in G$ and
a radius $r$, let $B_G(r,g)$ denote the open ball of radius $r$ around
$g \in G$ with respect to a suitable metric on $G$.  Then in general,
there are three possible standards for how well $A^\ell$ approximates $G$:
\begin{enumerate}
\item Existence of approximation: $A^\ell$ could be an $\eps$-net of $G$
for some $\eps > 0$, or an $\eps$-net of $B_G(R,1)$ for a large radius
$R$ if $G$ is not compact.

\item Algorithmic approximation: Given $g \in G$, there could be an
(polynomially) efficient algorithm to find $w \in A^\ell \cap B_G(\eps,g)$,
assuming efficient algorithms to compute the coordinates of $g$ and each
$a \in A$.

\item Statistical approximation, or equidistribution: There could be roughly
the same number of elements in $A^\ell \cap B_G(\eps,g)$ for all $g \in G$,
or all $g \in B_G(R,1)$ when $G$ is not compact.
\end{enumerate}
Either statistical or algorithmic approximation implies existence of
approximation for any given $\eps$, and any given $R$ when $G$ is not
compact.  Nonetheless, as far as we know, all current results
for existence of approximation either come with an algorithm, or also
establish statistical approximation.

Note that if $G$ has a nontrivial quotient group $H = G/N$ for some
closed, normal subgroup $N$, then $H$ is another connected Lie group,
and approximation in $G$ is at least as difficult as approximation in $H$.

Note also that if $\Gamma \subset G$ is a finitely generated, dense
subgroup, then all three approximation problems are equivalent up to
constant factors for all finite, symmetric sets $A \subset \Gamma$ that
exactly generate $\Gamma$.

For the rest of this subsection, let $\eps = 2^{-n}$.

If $G$ is a perfect Lie group, then a well-known volume argument shows that
$A^\ell$ can only be an $\eps$-net when $\ell = \Omega(n)$, with an extra
long-distance $\Omega(R)$ term if $G$ is not compact.  Note that a similar
bound holds for compressed words; see \Sec{ss:models}.  The case when $G$
is not perfect is less interesting, because then $G$ has a nontrivial,
connected abelian quotient $G_{\ab}$.  For the Solovay--Kitaev problem
in an abelian Lie group $G_{\ab}$, an analogous counting argument instead
implies that
\[ \ell = \Omega(\eps^{-\alpha}) = 2^{\alpha n)}
    \qquad \alpha = \frac{\dim(G_{\ab})}{|A|-1}. \]

As mentioned in the introduction, the prior best algorithmic result
when $G = \SU(d)$ and $A$ is symmetric and otherwise general establishes
$\ell = O(n^{3+\delta})$ using the KSV algorithm \cite[Thm.~8.5]{KSV:cqc}.
Problem A3.1 in Nielsen and Chuang \cite{NC:qcqi} claims that approximations
exist with length $\ell = \tO(n^2)$, but this exercise seems erroneous
and we do not know how to repair it.

For certain special gate sets and $G = \SU(2)$,
collaborators at Dalhousie, Microsoft, and Waterloo
\cite{KMM:practical,KBS:topological,BGS:vbasis,RS:optimal,KBRY:framework}
found efficient approximation algorithms with $\ell = O(n)$, which remarkably
is optimal up to a constant factor.  These algorithms use a number-theoretic
approach that does not seem to apply to general gate sets; they run in
randomized polynomial time modulo heuristic hypotheses in number theory.

Equidistribution is easier to consider when $G$ is a compact group.
As far as we know, all current statistical approximation results establish
a spectral gap for the operator average of elements of $A$ acting on a
unitary representation $V$ of $G$.  As explained by Harrow, Recht, and
Chuang \cite[Thm.~1]{HRC:discrete}, if $A$ has a uniform spectral gap,
independent of $V$, then there is a word in $A$ with optimal length
$\ell = O(n)$ to approximate any $g \in G$.   By a remark of Varju
\cite[Secs.~2.1\&3]{Varju:compact}, if $G$ is compact, then a spectral
gap result for all symmetric $A$ implies a spectral gap result for all
asymmetric $A$.

Spectral gap results of this type have a separate and longer history.
First, Kazhdan established that certain finitely generated discrete groups
$\Gamma$ have a universal spectral gap property, which he called property
(T) \cite{Kazhdan:dual}.  Later, Margulis \cite{Margulis:means}, Sullivan
\cite{Sullivan:additive}, Drinfeld \cite{Drinfeld:additive}, and Lubotzky,
Phillips, and Sarnak \cite{LPS:hecke1,LPS:hecke2} found examples of finite
sets $A$ in particular compact $G$ with uniform spectral gaps.  Lubotzky,
Phillips, and Sarnak also emphasized the connection between spectral gaps
and explicit bounds for equidistribution.  This work led to the conjecture
that any $A$ that densely generates any compact, semisimple $G$ has a uniform
spectral gap.  Building on work of Bourgain and Gamburd \cite{BG:su2,BG:sud},
Benoist and de Saxce \cite{BdS:gap} established a uniform spectral gap when
$G$ is simple and $A$ is an algebraic gate set.

Boutonnet, Ioana, and Salehi-Golsefidy \cite{BIG:local} recently generalized
the spectral gap condition to a local condition that makes sense for
noncompact groups and infinite measure spaces.  They also generalized the
theorem of Benoist and de Saxce to prove the following: If $G$ is a simple
real Lie group that need not be compact, then every finitely generated,
dense, algebraic group $\Gamma \subset G$ has a symmetric subset $A
\subseteq \Gamma$ with a local spectral gap.  Their Corollary H implies
a nonalgorithmic existence version of our result for every symmetric
algebraic gate set $A$ that densely generates $G$, with an optimal word
length of $\ell = O(n+R)$.

When $G$ is compact and semisimple and $A$ densely generates $G$ and is
otherwise arbitrary, Varju \cite[Sec.~1.5]{Varju:compact} established
equidistribution with words of length $\ell = O(n^3)$.  We wonder whether
our methods can be used to reduce this exponent.

Finally, in the inverse-free Solovay--Kitaev theorem of
Bouland--Giurgica-Tiron \cite{BG:ifsk} for $\SU(d)$ and $\SL(d,\C)$, the
word length exponent $\alpha$ increases with $d$.  It is an interesting
question whether there is a polynomial-time algorithm for this inverse-free
case with any bound for $\alpha$ independent of $d$.

\section{Preliminaries}
\label{s:prelim}

We will use asymptotic notation $g(x) = O(f(x))$ (and the companion
notations $\Omega(f(x))$ and $\Theta(f(x))$) with the standard meaning
that $|g(x)| < Cf(x)$ eventually as $x \to 0$, or as $x \to \infty$,
depending on context.  In some cases we will use the same notation in
unconditional or nonasymptotic form, meaning that $|g(x)| < Cf(x)$ for
all $x$.  In other cases, we use a multivariate variate version such as
$g(x,y) = O(f(x,y))$, which then means that $|g(x,y)| < Cf(x,y)$ when
$\min(x,y)$ is sufficiently large (or sufficiently small, \etc.)

We will also use the standard notation
\[ \tOmega(f(x)) \defeq \bigcup_{\alpha > 0} O(\log(f(x))^\alpha f(x)), \]
along with the analogous $\tOmega(f(x))$ and $\tTheta(f(x))$.  In particular,
a function $f(n) = \tTheta(n)$ is called \emph{almost linear}.

In the introduction, we also use the abbreviation $\poly(f(x))$, which
means the union of $O(f(x)^\alpha)$ for all $\alpha > 0$.

\subsection{Models of numbers, words, and computation}
\label{ss:models}

In this article, we will estimate the serial time complexity of
algorithms in the RAM machine model \cite{CR:random}, which is known to
be almost-linear equivalent to a Turing machine with an infinite regular
tree tape \cite{PR:versus}.  Models of serial computation are arguably
only stable up to almost-linear time equivalence, so that $\tO(f(n))$
is more robust than $O(f(n))$ as a time complexity bound in general.

We will use a version of interval arithmetic to input and calculate with
real numbers.  Explicitly, we represent a real number $x \in \R$ by a
dyadic rational number $x' \in \Q$ and allowed error $\eps \in \Q_+$ such
that $\abs{x-x'} < \eps$.  In general, we will not require that $\eps =
2^{-n}$; if $\eps \le 2^{-n}$, then we say that $x$ is given with $n$
bits of precision.  We do not require that $\eps$ be a power of $2$.
If $G \subseteq \M(d,\R)$ is a matrix group and $g \in G$, then we will
likewise represent it by $g' \in \GL(d,\R) \subseteq \M(d,\R)$ such that
\[ \max(\norm{g-g'}_\infty,\norm{g^{-1}-{g'}^{-1}}_\infty) < \eps, \]
where $\norm{\cdot}_\infty$ is the operator norm on $\M(d,\R)$.

Thanks to fast integer multiplication \cite{SS:schnelle} and other
techniques, various numerical primitives can be calculated in almost
linear time in the number of digits of precision.  First, we can convert
numbers between base 2 and base $b$ for any $b \ge 2$, or we can let an
approximation $x' \in \Q$ be any rational with the numerator and denominator
in any particular base $b$.  Thus, the interval arithmetic data type is base
independent, up to almost-linear-time reductions.  Also, using Newton's
method (with fast division, which itself comes from Newton's method),
we can also expand the digits of a nondegenerate, isolated solution to
a fixed system of polynomial equations.  This will be used in our algorithms.

Fast integer multiplication also supports the hypothesis of \Thm{th:runtime}
in a different way.  Many irrational numbers that arise in calculus can be
expanded in almost linear time \cite{BB:familiar,LMS:volume}, including
not only all algebraic numbers but also many transcendental values.
Thus we obtain the same time complexity estimates if the matrix entries
of the gates $a \in A$ or the target $g \in G$ are given in closed form
complicated proven or likely transcendental values such as $\sqrt{e+\pi}$.

If the matrix entries of the gates $a \in A$ are all algebraic numbers,
then as stated in \Thm{th:runtime}, we will use a second data type that
yields better time complexity.  In this case, the real components of the
gates (the real and imaginary parts of all of the matrix entries) lie in a
common real number field $K \subseteq \R$.  Given a rational basis of $K$,
we can describe any $x \in K$ as a vector of rational numbers.

A word $w$ over the gate set $A$ can either mean an element $w \in F_A$ in
the free group $F_A$ generated by $A$, with inverse letters identified as
inverse group elements; or the value $w \in G$ of the corresponding product
in the group $G$.  Since there is a canonical evaluation map from $F_A$,
we can take $w \in F_A$ and coerce it to an element of $G$ as needed. This
coercion is usually but not always implicit.  For example, if $A$ is an
algebraic gate set, then we can relate the word length $\len(w)$ of $w
\in F_A$ to the bit complexity of $w \in G$ by the useful estimate
\begin{eq}{e:bitlen} \norm{w_G}_{\bit} = O(\len(w_A)). \end{eq}

A \emph{compressed word} that represents a word $w$ over an alphabet $A$
is an algebraic circuit $\Delta$ that computes $w$ by combining letters
in $A$ via concatenation gates (with arbitrary fan-in) and unary inversion
gates \cite{Lohrey:groups}.  We can evaluate $\Delta$ and thus uncompress
$w$ in almost linear time $\tO(\len(w))$, provided that $O(\len(w))$
letters cancel during the evaluation.  However, a poorly behaved circuit
$\Delta$ may induce an exponential amount of cancellation.  We do not know
whether such a circuit $\Delta$ can be evaluated in almost linear time
in the size of $\Delta$ (although it can be evaluated in polynomial time
\cite[Sec.~4.4]{Lohrey:groups}).  We will sidestep this issue by bounding
the length of $w$ without cancelling letters whenever $w$ is the value of
a given circuit $\Delta$.

Finally, let $\len(\Delta)$ be a circuit $\Delta$ that computes $w$, which
we define as the total fan-in of the concatenation gates in $\Delta$.
By a standard circuit counting argument, there are $\exp(\tO(\hell))$
circuits of length at most $\hell$.  It follows that
\[ \hell = \len(\Delta) = \tOmega(n+R) \]
for most targets $g \in B_G(R,1)$ in the statement of \Thm{th:semisimple},
by the same volume argument that tells us that
\[ \ell = \len(w) = \Omega(n+R). \]

\subsection{Matrix Lie groups}
\label{ss:matrix}

Most of the material in this subsection is discussed in the textbook by
Hall \cite{Hall:gtm}.

An (explicit) \emph{real matrix Lie group} is a subgroup $G \subseteq
\GL(d,\R)$ of the group of invertible real $n \times n$ matrices for
some $d$, which is also a closed subset and a smooth submanifold.  Most,
but not quite all, of the strength of \Thm{th:semisimple} is realized
by considering matrix Lie groups, because there are also Lie groups that
cannot be realized with matrices.  (See \Sec{ss:lie} for a full discussion,
including the definition of a semisimple Lie group.)  For now, observe that
\[ \GL(d,\R) \subseteq \GL(d,\C) \into \GL(2d,\R), \]
which tells us that we can equivalently define real matrix Lie groups with
complex matrices.  A \emph{complex matrix Lie group} is a special case of
a real Lie group, in which $G \subseteq \GL(d,\C)$ as a manifold admits
complex analytic coordinates.

Even though our group elements are usually matrices, we will use lowercase
letters $g \in G$ for group elements and $1 \in G$ for the identity element.

\begin{examples} The groups $\U(d)$ and $\SU(d)$ are both real matrix
Lie groups.  The group $\PU(d) \cong \PSU(d)$ (unitary matrices quotiented
by scalars) is also a matrix group, if we represent each element $u
\in \PU(d)$ by its conjugation action on matrices, $x \mapsto uxu^{-1}$.

Likewise, the groups $\SL(d,\C)$ and $\PGL(d,\C) \cong \PSL(d,\C)$ are
both complex matrix Lie groups.
\end{examples}

If $G$ is a matrix Lie group, then its tangent space $L = T_1(G)$ at the
identity is called the \emph{Lie algebra} of $G$.  The Lie algebra $L$
is closed under the algebra commutator
\[ [x,y] \defeq xy-yx, \]
which in context is called the \emph{Lie bracket} of $L$.  Also, the matrix
exponentiation map $\exp:L \to G$ sends $L$ into $G$ and covers an open
neighborhood of the identity $1 \in G$, although it does not always reach
all of $G$.

\begin{examples} The Lie algebras $\sl(d,\C)$, $\u(d)$, and $\su(d)$
of $\SL(d,\C)$, $\U(d)$, and $\SU(d)$ consist of traceless matrices,
anti-Hermitian matrices, and matrices that are both traceless and
anti-Hermitian, respectively.  \end{examples}

For any group $G$, we will use these abbreviations for conjugation and
the group commutator:
\[ g^h \defeq hgh^{-1} \qquad \comm{g,h} \defeq ghg^{-1}h^{-1}. \]
If $G$ is a matrix Lie group, then
\begin{multline}
\comm{\exp(x),\exp(y)} \\
= \exp\bigl([x,y]
    + O(\norm{x}^2\,\norm{y}+\norm{x}\,\norm{y}^2)\bigr).
\label{e:glbracket} \end{multline}
\Equ{e:glbracket} generalizes to the Baker--Campbell--Hausdorff (BCH)
formula:
\begin{multline} \exp(x)\exp(y) = \\
    \exp\biggl({\displaystyle x+y + \frac{[x,y]}2 + \frac{[x,[x,y]]
    + [y,[y,x]]}{12} + \cdots} \biggr)
\label{e:bch} \end{multline}
The exponent in \equ{e:bch} is expressed entirely in terms of the Lie
bracket of $L$.  The power series has a positive radius of convergence in
both $\norm{x}$ and $\norm{y}$, and thus yields a real analytic formula
for the group law of $G$ in a open neighborhood of $1$.

\section{The algorithm for $\SU(2)$}
\label{s:su2}

In this section, we will prove \Thm{th:qudit} in the qubit case $G = \SU(2)$.
The algorithm in the full generality of \Thm{th:semisimple} is similar, but
it is simpler and clearer in this special case.  Throughout this section,
we assume a finite gate set or generating set
\[ 1 \in A = A^{-1} \subseteq \SU(2) \]
that generates a dense subgroup of $\SU(2)$.  The algorithms are also
simpler if we replace the integer parameter $n$ by a rational parameter $t
\ge 1$ in the statements of Theorems~\ref{th:qudit}, \ref{th:semisimple},
and \ref{th:runtime}.

\subsection{Geometry}
\label{ss:su2geom}

We review some geometric and algebraic facts about the Lie group $\SU(2)$.

Using the standard Pauli matrices,
\[ X = \begin{bmatrix} 0 & 1 \\ 1 & 0 \end{bmatrix} \quad
Y = \begin{bmatrix} 0 & -i \\ i & 0 \end{bmatrix} \quad
Z = \begin{bmatrix} 1 & 0 \\ 0 & -1 \end{bmatrix},
\]
we can represent the quaternions $\H$ with $2 \times 2$ complex matrices
and interpret $\SU(2)$ as the unit 3-sphere in $\H$:
\begin{align*}
\H &\cong \{aI+ibX+icY+idZ \st a,b,c,d \in \R\} \\
\SU(2) &= \{aI+ibX+icY+idZ \st a^2 + b^2 + c^2 + d^2 = 1\}.
\end{align*}
We can then use spherical geometry to define an angular distance between
elements two elements $g,h \in \SU(2)$:
\[ d(g,h) = d_{\SU(2)}{(g,h)} \defeq \arccos \frac{\tr(gh^{-1})}{2}. \]
This distance is invariant under both left and right multiplication in
$\SU(2)$, or \emph{bi-invariant}:
\[ d(gh,h) = d(hg,h) = d(g,1) \qquad d(g^h,1) = d(g,1). \]
In other words, if $g \sim h$ are conjugate elements in a group $G$ with a
bi-invariant metric, then $d(g,1) = d(h,1)$.  In $\SU(2)$ with its spherical
metric, more is true: $d(g,1) = d(h,1)$ if and only if $g \sim h$.

Since the metric is bi-invariant, metric balls
\[ B(r,g) = B_{\SU(2)}(r,g) \]
translate to each other using either left or right multiplication:
\[ B(r,g) = gB(r,1) = B(r,1)g. \]
Using this, we obtain a set arithmetic relation that will be useful for
numerical precision estimates:
\begin{eq}{e:bbmul} B(\eps,g)B(\delta,h) \subseteq B(\eps+\delta,gh).
\end{eq}

Finally, we define the angle $\angle g,h$ between two elements $g,h \in
\SU(2)$ to be the angle at the corner $1 \in \SU(2)$ of the spherical
triangle with vertices at $1$, $g$, and $h$.  We will use this elementary
proposition, which we leave as an exercise in quaternion arithmetic:

\begin{proposition} If $g,h \in \SU(2)$ with
\[ g \sim h \qquad d(g,1) = d(h,1) = \psi \]
and $\angle\, g,h = \theta$, then
\begin{align*}
d(\comm{g,h},1) &= 2\arcsin \bigl( \sin(\psi)^2\sin(\theta) \bigr) \\
    &= (2\psi^2+O(\psi^4))\sin(\theta)
\end{align*}
as $\psi \to 0$, uniformly in $\theta$.
\label{p:cross} \end{proposition}

\subsection{Roughly exponential steps}
\label{ss:steps}

The first part of the algorithm constructs an $A$-word $s_t$ for any given
rational number $t > 0$ such that
\begin{eq}{e:step} 2^{1-t} > d(s_t,1) > 2^{-t}
    \qquad \len(s_t) = O(t^\alpha) \end{eq}
for some favorable value of $\alpha > 1$.  We call such a word (and its
value) a \emph{roughly exponential step}.  Although $s_t$ is required to
have a favorable word length and a somewhat favorable distance, it does
not need to point in any particular direction, and it is unrelated to any
target element $g \in \SU(2)$.  There is also nothing special about powers
of 2 in \eqref{e:step}; it would suffice to require any constraint
of the form $d(s_t,1) = \Theta(b^{-t})$ with any $b > 1$.

We start with the case $\alpha = 2$, using standard group commutators.

\begin{algorithm}{SB2} The input is a finite gate set
\[ A = A^{-1} \subseteq \SU(2) \]
that densely generates $\SU(2)$ and a rational target precision $t > 0$.
The algorithm also depends on two integer constants $a,b > 0$ (depending only
on $A$) that must be chosen favorably.

The output is an $A$-word $s_t$ that satisfies the distance condition
in \eqref{e:step}.
\begin{enumerate}
\item If $t \le 2a$, search for an $A$-word $s_t$ with $\len(s_t) \le b$ that
satisfies the distance condition of \eqref{e:step}, and return this word.
\item If $t > 2a$, recursively (using this algorithm) calculate an
$A$-word $s_{t'}$ with
\[ t' := \frac{t-a}2. \]
Then search for an $A$-word $u$ such that $\len(u) \le b$, and such that
\begin{eq}{e:sb2} s_t := \comm{s_{t'},s_{t'}^u} \end{eq}
satisfies \equ{e:step}.  Return $s_t$.
\end{enumerate}
\label{a:sb2} \end{algorithm}

Step 2 in \Alg{a:sb2} can succeed if the constant $a$ is large enough,
and if $b$ is large enough for a given value of $a$.  Setting
\begin{eq}{e:ghpsi} g = s_{t'} \qquad h = s_{t'}^u
    \qquad \psi = d(s_{t'},1), \end{eq}
the exact distance formula in \Prop{p:cross} implies that $a=1$ suffices.
Also, working only from the asymptotic expression, we can choose $a$
large enough to guarantee that
\[ \psi^2 \sin(\theta) d(\comm{s_{t'},s_{t'}^u}) < 3\psi^2 \sin(\theta). \]
If $a$ is suitably chosen, then in any given instance of step 2 we
can choose $u$ to control the angle $\theta = \angle\,s_{t'},s_{t'}^u$
in \Prop{p:cross}.  Under the hypotheses, $\theta$ only needs a bounded
amount of precision to put $d(s_t,1)$ in the range required by \equ{e:step}.
Thus, $u_{\SU}(2)$ only needs a bounded amount of precision, so we can bound
$\len(u_A)$ by some constant $b$.

If $a$ and $b$ are both suitably chosen, \equ{e:sb2} yields the recurrence
relation
\[ \len(s_t) \le 4(\len(s_{(t-a)/2}) + b). \]
This recurrence implies the word length bound in \eqref{e:step}.
The structure of \Alg{a:sb2}, in particular \equ{e:sb2}, also yields a
compressed word $\Delta_t$ form of $s_t$ that satisfies the recurrence
\[ \len(\Delta_t) \le \len(\Delta_{(t-a)/2}) + b. \]
This second recurrence implies that
\begin{eq}{e:ctlen} \len(\Delta_t) = O(\log(t)) = \tO(1). \end{eq}

We also estimate the time complexity of \Alg{a:sb2}.  As the estimate
for $\len(\Delta_t)$ indicates, the algorithm uses $\tO(1)$ arithmetic
and logical steps.  If we only wanted the compressed word form of $s_t$,
namely $\Delta_t$, then the algorithm could be implemented with the steps
kept in Lie algebra form $\log(s_t)$ with only logarithmic floating-point
precision, and thus run in time $\tO(1)$ in total.  However, the algorithm
is only useful as a subroutine for \Alg{a:zb} later if $(s_t)_{\SU(2)}$
is calculated to $t+\Theta(t)$ bits of precision.

We claim that the gates in $A$ and all other group elements in the algorithm
only need $t+\Theta(t)$ bits of precision for this outcome.  To see this,
it follows from \equ{e:bbmul} that the bits of precision of the value of
an arithmetic circuit in $\SU(2)$ diminish in proportion to its depth.
In our case, the circuit $\Delta_t$ has depth $O(\log(t))$ as well as
size $O(\log(t))$.  The precision claim tells us that the algorithm runs
in $\tO(t)$ time if elements of $\SU(2)$ are modeled with generalized
interval arithmetic, and we use fast multiplication to calculate products.

If the gate set $A$ is algebraic with entries in a number field $K$, then we
instead express each step $s_t$ in exact algebraic form.  By \equ{e:bitlen},
it follows that $\norm{s_t}_{\bit} = O(t^2)$, and that the time complexity
of \Alg{a:sb2} is $\tO(t^2)$.  This is worse than the time complexity
using interval arithmetic, but it will lead to a better result in \Alg{a:zb}.

\Alg{a:sb} below is a generalized version of \Alg{a:sb2} that uses higher
commutators to achieve a better exponent $\alpha$.  We will need some
extra definitions and background results for the generalization.

Let $\omega \in F_k$ be a word in the free group on $k$ letters,
interpreted as a group word that can be evaluated in any group
$G$. If $G$ is a connected matrix Lie group with Lie algebra $L$,
then we define the \emph{cancellation degree}
to be the largest integer $n$ (if there is a largest one) such that
\begin{eq}{e:candeg} \omega(\exp(\eps x_1),\ldots,\exp(\eps x_k))
    = \exp(O(\eps^n)) \end{eq}
for all $x_1,\ldots,x_k \in L$.  If $\omega(\vg) = 1$ identically on
$G^k$, then by default we let $\can_G(\omega) = \infty$.  For example,
the group commutator $\omega(g,h)= \comm{g,h}$ has cancellation degree 2
when $G$ is nonabelian.

We will need a variation of this concept, the \emph{conjugate
cancellation degree} $\ccan_G(\omega)$, which is given by the same formula
\eqref{e:candeg}, but only when the group arguments (equivalently, the
Lie algebra arguments) are conjugate.  If $\can_G(\omega) \ge 3$ for
all $G$, then $\omega$ can be thought of as a \emph{higher commutator}.
(See also \Sec{ss:cannil} and in particular \Thm{th:cannil}.)

A word $\omega \in F_k$ of length $\len(\omega)$ then has a \emph{conjugate
cancellation exponent}
\[ \alpha_G(\omega) \defeq \frac{\log(\len(\omega))}
    {\log(\ccan_G{(\omega))}}. \]
As explained in \Sec{ss:cannil}, Elkasapy and Thom
\cite{ET:lower,Elkasapy:lower} discovered that there are higher commutators
with a lower exponent (and are thus more efficient) than the usual group
commutator $\comm{g,h}$.   Their first observation is that
\begin{align}
\omega(f,g,h) &= \bigcomm{\comm{f,g},\comm{g,h}} \nonumber \\
    &= fgf^{-1}g^{-1} \cdot ghg^{-1}h^{-1} \cdot
        gfg^{-1}f^{-1} \cdot hgh^{-1}g^{-1} \nonumber \\
    &= fgf^{-1}hg^{-1}h^{-1}gfg^{-1}f^{-1}hgh^{-1}g^{-1}
\label{e:len14} \end{align}
has length 14 rather than 16, which yields
\[ \alpha_{\SU(2)}(\omega) = \log_4(14) = 1.90367\ldots. \]
As explained in \Sec{s:elkasapy}, Elkasapy later found a sequence of words
$\omega_n(g,h)$ such that
\[ \lim_{n \to \infty} \frac{\log(\len(\omega_n))}
    {\log(\ccan_{\SU(2)}(\omega_n))} = \log_\phi(2). \]

\begin{lemma} Let $\omega(g,h)$ be a nontrivial group word with
\[ \ccan_{\SU(2)}(\omega) = n \ge 2. \]
Let $g,h \in \SU(2)$ with
\[ d(g,1) = d(h,1) = \psi \qquad \angle g,h = \theta.\]
Then
\[ \cos(d(\omega(g,h),1)) = p(\psi,\theta) \]
for some trigonometric polynomial $p(\psi,\theta)$ such
that $p(\psi,0) = 0$.  In addition,
\[ d(\omega(g,h),1) = \psi^n \sqrt{q(\theta)} + O(\psi^{n+2}) \]
as $\psi \to 0$, where $q(\theta)$ is a trigonometric polynomial such that
$q(0) = 0$.
\label{l:trig} \end{lemma}

\begin{proof} We can assume that
\[ g = exp(i\psi Z) \qquad h = \exp(i\psi Z)^{\exp(i\theta X)}, \]
since the hypotheses imply this form up to conjugation.  Then every real
component of the matrix $\omega(g,h)$ is a trigonometric polynomial in
$\psi$ and $\theta$.  The same is true of
\[ p(\psi,\theta) = \frac{\tr(\omega(g,h))}2, \]
except in that it must specifically be a real polynomial, since the trace
of any element of $\SU(2)$ is real.

For the second claim, first note that
\[ p(-\psi,\theta) = p(\psi,\theta), \]
since $g^{-1}$ and $h^{-1}$ have the same properties as $g$ and $h$.
Thus the Taylor series expansion of $p(\psi,\theta)$ only uses even
powers of $\psi$; in particular,
\[ p(\psi,\theta) = 1 - \frac{\psi^{2n} q(\theta)}2
    + O(\psi^{2n+2}) \]
for some integer $n > 0$.  Applying $\arccos$, we get
\[ d(\omega(g,h),1) = \psi^n \sqrt{q(\theta)} + O(\psi^{n+2}). \]
Since $g$ and $h$ are conjugate if and only if $d(g,1) = d(h,1)$, the
exponent $n$ amounts to the definition of $\ccan_{\SU(2)}(\omega)$.
Finally, $p(\psi,0) = 0$ and $q(0) = 0$ because $\omega(g,h)$ has
to have trivial abelianization in order to have any cancellation, which
implies that $\omega(g,g) = 0$.
\end{proof}

Here is our more general algorithm to generate roughly exponential steps
for other values of $\alpha > 1$.

\begin{algorithm}{SB} The input is a finite set
\[ A = A^{-1} \subseteq \SU(2) \]
that densely generates $\SU(2)$, a group word $\omega(g,h)$ with
\[  n = \ccan_{\SU(2)}(\omega) \ge 2, \]
and a rational target precision $t > 0$.  The algorithm also depends
on two integer constants $a,b > 0$ (depending only on $A$) that must be
chosen favorably.

The output is an $A$-word $s_t$ that satisfies the distance condition
in \eqref{e:step}.
\begin{enumerate}
\item If $t \le 2a$, search for an $A$-word $s_t$ with $\len(s_t) \le b$ that
satisfies the distance condition of \eqref{e:step}, and return this word.
\item If $t > 2a$, recursively (using this algorithm) calculate an
$A$-word $s_{t'}$ with
\[ t' := \frac{t-a}n. \]
Then search for an $A$-word $u$ such that $\len(u) \le b$, and such that
\[ s_t := \omega(s_{t'},s_{t'}^u) \]
satisfies the distance condition in \equ{e:step}.
\end{enumerate}
\label{a:sb} \end{algorithm}

Like \Alg{a:sb2}, step 2 in \Alg{a:sb} can succeed if the constant $a$ is
large enough, and if $b$ is large enough for a given value of $a$.  We can
relate \Alg{a:sb} to \Lem{l:trig} with the substitutions \eqref{e:ghpsi}.
First, when $a$ is large enough $d(\omega(g,h),1)$ is very close to
$\psi^n\sqrt{q(\theta)}$.  Second, if we let $m$ be the maximum value of
$\sqrt{q(\theta)}$, then we also want
\[ m2^{-ct'} > 2^{-t} = 2^{1-ct'-a}. \]
It suffices to also take $a > \log_2(m)$.  If $a$ is suitably chosen, then
in any given instance of step 2 we can choose $u$ to control the angle
$\theta$ in \Lem{l:trig}.  Under the hypotheses, $\theta$ and therefore
$u_{\SU(2)}$ only need bounded precision to satisfy \equ{e:step}, so we
can bound $\len(u_A)$ by some constant $b$.

For the word length analysis, let
\begin{eq}{e:elal} \quad \ell = \len(\omega)
    \quad \alpha = \log_n(\ell), \end{eq}
and assume that $\alpha > 1$.  We can assume for convenience and without
loss of generality that at most half of the letters in $\omega(g,h)$
are $h$ or $h^{-1}$.  If so, and if $a$ and $b$ are suitably chosen,
then the length recurrence for $s_t$ in \Alg{a:sb} is
\begin{eq}{e:lenrec} \len(s_t) \le \ell(\len(s_{(t-a)/n}) + b). \end{eq}
This recurrence implies the word length bound in \eqref{e:step} with the
given value of $\alpha$.  \Alg{a:sb} also yields a compressed word $\Delta_t$
for $s_t$ that satisfies the recurrence
\begin{eq}{e:clenrec} \len(\Delta_t) \le \len(\Delta_{(t-a)/n})+b, \end{eq}
which again implies the length bound \eqref{e:ctlen}.

The time complexity and precision demands of \Alg{a:sb} are also
analogous to that of \Alg{a:sb2}.  If the algorithm is implemented with
generalized interval arithmetic, then all of the group elements in $\SU(2)$
that it uses need $n+\Theta(n)$ bits of precision, and the algorithm runs
in time $\tO(n)$.  If instead the algorithm is implemented with exact
arithmetic over a number field $K$, then $\norm{s_t}_{\bit} = O(t^\alpha)$
by \equ{e:bitlen}, and the algorithm runs in time $\tO(t^\alpha)$.

\subsection{Zigzag golf}
\label{ss:zgolf}

In this subsection, we establish \Thm{th:qudit} in the $d=2$ case of a qubit.
Given  $g \in \SU(2)$, \Alg{a:zb} produces an approximation $w_t \sim g$
such that
\begin{eq}{e:wtest} d(w_t,g) < 2^{-t} \qquad \len(w_t) = O(t^\alpha) \end{eq}
for any $\alpha > \log_\phi(2)$.  The algorithm employs the roughly
exponential steps $\{s_t\}$ provided by \Alg{a:sb} in a strategy that
we call \emph{zigzag golf}.  Even though the distance $d(s_t,1)$ is not
all that precise, it is a fixed function of $t$.  We can also control
the direction of $s_t$ by conjugating it with another $A$-word $u$ to
make $s_t^u$.  Given $g \in \SU(2)$, if we construct each approximation
$w_t$ as a product of steps of the form $s_{t'}^u$, then we can interpret
each $s_{t'}^u$ as a golf stroke that moves a golf ball in $\SU(2)$ by a
rigid distance and in a precise direction.  A product $s_{t'}^us_{t'}^v$
of two golf strokes is then a zigzag that can reach a precise distance
and direction.  (Note also that $u$ and $v$ depend implicitly on $t'$.)
Thus $w_t = s_{t'}^us_{t'}^vw_{t'}$ can be a much more precise approximation
to $g$ than the previous $w_{t'}$.  See \Fig{f:zgolf}.

\begin{figure}
\begin{ctp}[thick]
\draw (0,0) circle (4);
\draw[dashed] (4,0) arc (0:180:4 and .95);
\draw (-4,0) arc (180:360:4 and .95);
\draw[darkgreen,dashed,rotate=105.4]
    (-1.40,2.46) arc (142.1:198.7:1.772 and 4);
\draw[darkgreen,dashed,rotate=127.8]
    (-2.040,-0.544) arc (-172.2:-146.4:2.060 and 4);
\draw[draw=darkred,rotate=156.8]
    (1.052,2.628) arc (139.0:194.5:-1.392 and 4)
    node[left,pos=.65,inner sep=2pt] {$s_1^{u_1}$};
\draw[draw=darkred,rotate=43.1]
    (0.376,1.636) arc (155.8:211.4:-0.412 and 4)
    node[below left,pos=.5,inner sep=1pt] {$s_1^{v_1}$};
\draw[draw=darkred,rotate=-169.8]
    (-1.428,1.556) arc (157.1:179.1:1.548 and 4)
    node[left,pos=.7,inner sep=2.3pt] {\small $s_2^{u_2}$};
\draw[draw=darkred,rotate=84.2]
    (0.368,-1.508) arc (-157.9:-135.9:-0.396 and 4)
    node[above,pos=.5,inner sep=3pt] {\small $s_2^{v_2}$};
\draw[darkblue,dashed,rotate=114.5]
    (-0.992,2.648) arc (138.5:221.0:1.32 and 4);
\draw[->,gray,rotate=170.3] (1.632,2.308) arc (144.7:162.2:-2 and 4);
\draw[->,gray,rotate=142.1] (0.348,2.808) arc (135.4:152.9:-0.492 and 4);
\draw[->,gray,rotate=127.9] (-0.348,2.808) arc (135.4:152.9:0.492 and 4);
\draw[->,gray,rotate=72.2] (-2.516,1.296) arc (161.1:178.6:2.656 and 4);
\fill (-2,-2) circle (0.07) node[below,inner sep=3.5pt] {\small 1};
\draw (-2.56,-1.68) node {\includegraphics[height=1cm]{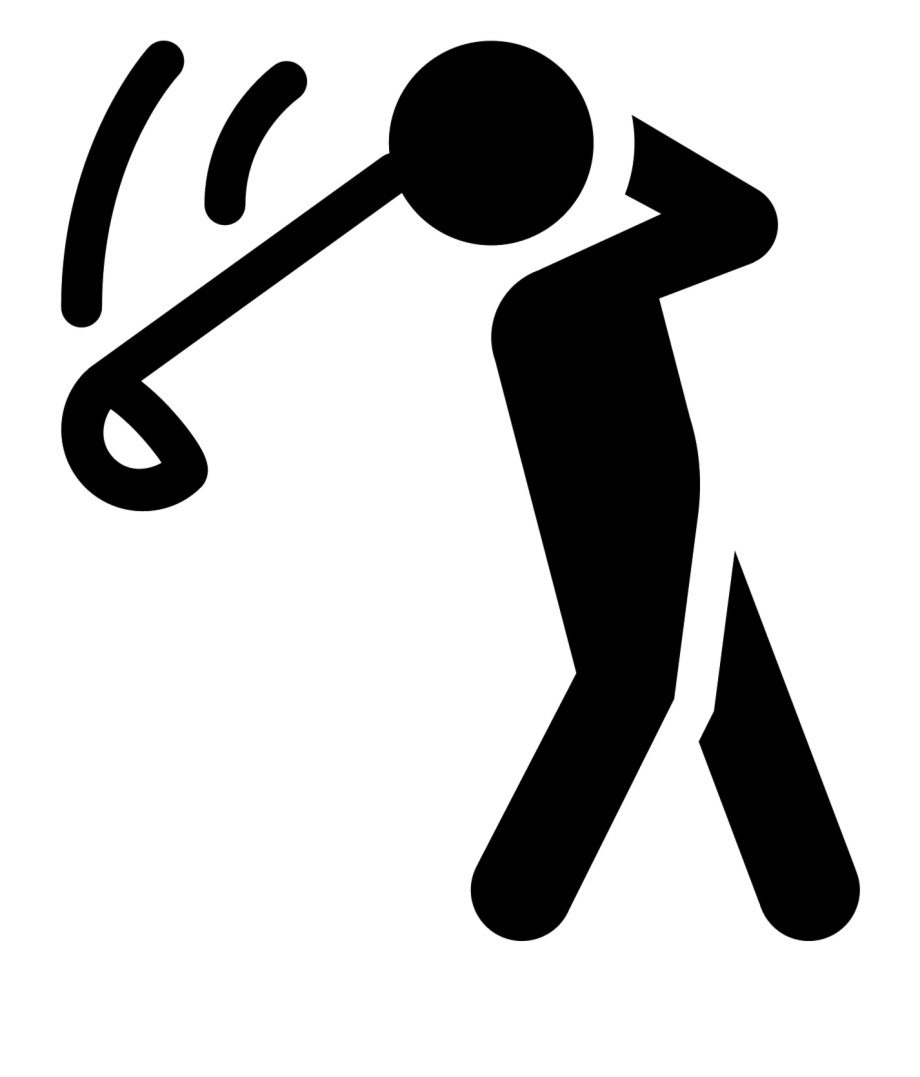}};
\fill (1.68,-1.28) circle (0.07) node[below right,inner sep=2pt] {$w_1$};
\fill (2.8,0) circle (0.07) node[right,inner sep=2.5pt] {$w_2$};
\fill (2.8,0.18) circle (0.07) node[above right,inner sep=2pt] {$g$};
\fill[darkred] (2.8,0.95) -- (3.05,0.85) -- (2.8,0.75);
\draw (2.8,0.18) -- (2.8,0.975);
\fill (-0.844,1.452) circle (0.07);
\fill (1.536,0.212) circle (0.07);
\draw (1.4,2.4) node {\large $\SU(2)$};
\end{ctp}
\caption{Zigzag golf with doubly exponential convergence. The steps $s_k$
    are each a fixed length, but they are aimed accurately using an earlier
    stage of the entire algorithm.  (The numbering here is simplified and
    differs from the main text.)}
\label{f:zgolf}
\end{figure}

To make the conjugators $u$ and $v$ precise, we construct them using at
an earlier stage the entire approximation algorithm to achieve doubly
exponential convergence.  Although this may seem expensive in word length,
we can adjust the parameters to achieve the same length exponent $\alpha$
for the total approximation $w_t$ as for an individual step $s_t$.

\begin{algorithm}{ZB} The input is a finite set
\[ A = A^{-1} \subseteq \SU(2) \]
that densely generates $\SU(2)$, a group word $\omega(g,h)$ with
$\ccan_{\SU(2)}(\omega) \ge 2$, a target element $g \in \SU(2)$, and a
rational target precision $t > 0$.  The algorithm also depends on a rational
constant $0 < \beta < 1$ (depending only on $A$ and $\omega$) and two integer
constants $a,b > 0$ (depending only on $A$) that must be chosen favorably.

The output is a word $w_t$ that satisfies \eqref{e:wtest}.
\begin{enumerate}
\item If $t(1-\beta) \le 2a$, search for an $A$-word $w_t$ with $\len(w_t)
\le b$ and $d(w_t,g) < 2^{-t}$, and return this word. If $t(1-\beta)
> 2a$, then proceed to the remaining steps.
\item Let $t' := (1-\beta)t$, and recursively calculate $w_{t'}$ using
this algorithm with input $(g,t')$.  Also calculate $s = s_{t'}$ using
\Alg{a:sb} and the word $\omega$.
\item Solve the equation
\[ s^{g_u} s^{g_v} = w_{t'}^{-1} g \]
for $g_u,g_v \in \SU(2)$.
\item Let $t'' := \beta t + a$, and calculate $u$ and $v$ recursively as
\[ u := w_{u,t''} \qquad v := w_{v,t''} \]
using $(g_u,t'')$ and $(g_v,t'')$ as inputs to this algorithm.  Return
\[ w_t := w_{t'} s^u s^v. \]
\end{enumerate}
\label{a:zb} \end{algorithm}

The rationale of step 4 in \Alg{a:zb} is that $u_{\SU(2)}$ and $v_{\SU(2)}$
only need $(t-t')+O(1) = \beta t + O(1)$ bits of precision, because
the conjugation map $u \mapsto s^u$ is Lipschitz with Lipschitz constant
$O(d(s,1)) = O(2^{-t'})$.  If $a$ is large enough, then $u$ and $v$ are close
enough to $g_u$ and $g_v$ to guarantee that $w_t$ satisfies the distance
bound in \eqref{e:wtest}.   Step 2 also guarantees that $d(w_{t'},g) <
d(s_{t'},1)$, which implies that zigzag from $w_{t'}$ reaches $g$ with
an acute triangle.  Thus the equation in step 3 is nonsingular, and the
approximation problem in step 4 is well-conditioned.  Note finally that
the conditional in step 1 ensures that $t'' \le t-a$ in step 4.

The length estimate in \eqref{e:wtest} is more delicate.  Assuming that
$\alpha > 1$, we claim that \eqref{e:wtest} holds when $\beta > 0$ is small
enough.  Informally, if we have faith that $\len(w_{t'}) = O((t')^\alpha)$
when $t' \le \beta t$, then the rate of the doubly exponential convergence
is proportional to $\beta$, while the word length of the conjugators $u$
and $v$ has a factor of $\beta^\alpha$. This indicates that $u$ and $v$
become ever more economical as $\beta \to 0$.  If $\beta$ is small enough,
then each $u$ and $v$ should contribute less to the length of $w_t$ than
the step $s$ does.  Although this informal argument smacks of circular
reasoning, it is ultimately valid.  To justify it, we give a more formal
argument with a specific value of $\beta$.

\begin{lemma} Let $\alpha > 1$ be the word length
exponent in \Alg{a:sb}, and let
\[ \beta = 4^{1/(1-\alpha)} \qquad D = \frac{3}{(1-\beta)^{1-\alpha}-1}
    \qquad t_0 \ge 1. \]
Given these parameters, choose $N$ such that \Alg{a:sb} produces a word $s_t$
with $\len(s_t) \le Nt^\alpha$ for all $t \ge 1$, and such that \Alg{a:zb}
produces a word $w_t$ with $\len(w_t) \le DNt^\alpha$ when $t_0 \ge t
\ge 1$. If $t_0$ is sufficiently large, then it follows that $\len(w_t)
\le DNt^\alpha$ for all $t \ge 1$.
\label{l:zb} \end{lemma}

Note the following in the statement of \Lem{l:zb}:
\begin{itemize}
\item Since $D > 0$, the constant $N$ can be chosen as specified.
\item Since $0 < \beta < 1$, the constant $\beta$ is valid in the algorithm.
\item The lemma states that $\len(w_t) = O(t^\alpha)$.
\end{itemize}

\begin{proof} For all $t \ge 2a/(1-\beta)$, \Alg{a:zb} produces:
\begin{gather*}
w_t = w_{t'} s_{t'}^u s_{t'}^v \qquad u = w_{u,t''}
    \qquad v = w_{v,t''} \\
t' = (1-\beta)t \qquad t'' = \beta t + a.
\end{gather*}
We now proceed by induction on $t$ when
\[ t > t_0 > \frac{2a}{1-\beta}. \]
Even though $t$ is rational and not an integer, we can still argue by
induction from the previous cases $t'$ and $t''$, given that
\[ t-t' = \beta t \ge 1 \qquad t-t'' \ge a \]
are both bounded below by positive constants.  We obtain:
\begin{align*}
\len(w_t) &= \len(w_{t'}) + 2\len(s_{t'}) + 2\len(u) + 2\len(v) \\
    &\le DN((1-\beta)t)^\alpha + 2N((1-\beta)t)^\alpha
        + 4DN(\beta t + a)^\alpha.
\end{align*}
We want to show that $\len(w_t) \le DNt^\alpha$.
Dividing through by $Nt^\alpha$, we want to show that
\[ D \ge D(1-\beta)^\alpha + 2(1-\beta)^\alpha
    + 4D\beta^\alpha + O(t^{-1}) \]
when $t$ is large enough.  If we substitute $\beta = 4\beta^\alpha$ and
move all terms with $D$ to the left side, we want to show that
\[ D(1-\beta - (1-\beta)^\alpha) \ge 2(1-\beta)^\alpha + O(t^{-1}). \]
Substituting the value of $D$, we want to show that
\[ 3(1-\beta)^\alpha \ge 2(1-\beta)^\alpha + O(t^{-1}), \]
which certainly holds when $t > t_0 \gg 1$.
\end{proof}

\begin{remark} In the simpler case $\alpha = 2$ provided by \Alg{a:sb2},
\Lem{l:zb} simplifies to $\beta = 1/4$ and $D = 6$.
\end{remark}

Just as \Alg{a:sb} produces a compressed word $C_{s,t}$ for the step
$s_t$, \Alg{a:zb} produces a compressed word $C_{w,t}$ for $w_t$.  We also
bound the compressed length and bound the time complexity of \Alg{a:zb}.
For this purpose we prove an extension of \Lem{l:zb}:

\begin{lemma} In the notation of \Lem{l:zb}, let $0 < \beta' < \beta$.
If \Alg{a:zb} uses the parameter $\beta'$ instead of $\beta$, then we
still obtain $\len(w_t) = O(t^\alpha)$.
\label{l:zbp} \end{lemma}

\begin{proof} Choose $\alpha'$ so
\[ \beta' = 4^{1/(1-\alpha')} \qquad 1 < \alpha' < \alpha, \]
and define modified lengths
\[ \len'(w_t) \defeq t^{\alpha'-\alpha} \len(w_t) \qquad
    \len'(s_t) \defeq t^{\alpha'-\alpha} \len(s_t). \]
Then the derivation in \Lem{l:zb} also holds with $\beta'$
and the modified lengths, to conclude that
\[ \len'(w_t) = O(t^{\alpha'}) \quad \implies
    \quad \len(w_t) = O(t^{\alpha}). \qedhere \]
\end{proof}

The compressed word $\Delta_{w,t}$ satisfies the length recurrence
\[ \len(\Delta_{w,t}) = \len(\Delta_{w,t'}) + \len(\Delta_{s,t})
    + \len(\Delta_{u,t''}) + \len(\Delta_{v,t''}), \]
which has smaller coefficients than the recurrence for $\len(w_t)$
in the proof of \Lem{l:zb}. Recall also that
\[ \len(\Delta_{s,t}) = \tO(1) = O(t^{\alpha'})
    \qquad \forall \alpha' > 1. \]
The two relations together let us apply \Lem{l:zb}
to the compressed length $\len(\Delta_{w,t})$ instead using
an arbitrarily small value $\beta' > 0$, to conclude that
\[ \len(\Delta_{w,t}) = O(t^{\alpha'}) = O(t^{1+\delta}) \]
for any $\alpha' > 1$, equivalently any $\delta > 0$.  This completes the
proof of \Thm{th:qudit} for $\SU(2)$.

To complete the proof of \Thm{th:runtime} for $G = \SU(2)$, if we execute
\Alg{a:zb} with generalized interval arithmetic, then its time complexity
has a similar analysis to that of \Alg{a:sb}.  Like $\Delta_{s,t}$, the
circuit $\Delta_{w,t}$ also has logarithmic depth, so the gates in $A$
need $t+O(\log(t))$ bits of precision.  The time complexity of \Alg{a:zb}
is also dominated by high precision arithmetic and is bounded by
\[ \len(\Delta_{w,t})\tO(t) = \tO(t^{2+\delta}) \]
for any $\delta > 0$.

Finally, if the gate set $A$ is algebraic and we calculate the exact values
of words over a real number field $K$, then \equ{e:bitlen} (together with
almost-linear-time arithmetic) implies that the time complexity of \Alg{a:zb}
is $\tO(t^\alpha)$, by essentially the same estimates as in \Lem{l:zb}.

\section{Background on Lie groups}
\label{s:backlie}

The reader who is interested in \Thm{th:semisimple} when $G = \SU(d)$,
or more generally when $G$ is compact and semisimple, can mostly skip
this section.

\subsection{Lie groups}
\label{ss:lie}

Much of the material in this subsection can be found in textbooks
such as the one by Varadarajan \cite{Varadarajan:gtm} and Helgason
\cite{Helgason:symmetric}.  In particular, Helgason not only treats the
classification of compact and complex simple Lie groups (which Hall
and Varadarajan also do), but also discusses Cartan decompositions,
symmetric spaces, and Hermitian type.  Section X.6 in Helgason describes
the simple real Lie algebras and a Lie group for each one.  See also Tits
\cite{Tits:tabellen}.

By definition, a \emph{real Lie group} is a set $G$ that is both a group
and a smooth manifold with real coordinates, such that the group and
inverse laws of $G$ are both smooth maps.  Every Lie group $G$ admits
compatible real analytic coordinates, and a \emph{complex Lie group} is
likewise a complex analytic manifold with a complex analytic group law.
(Note that every complex Lie group can also be viewed as a real Lie group.)
In this article, we will usually work in the generality of semisimple Lie
groups (defined below) which are connected and either real or complex.
We will usually drop the adjectives ``connected'' and ``real''.

For any commutative ring $A$, let $\M(d,A)$ be the algebra of $d
\times d$ matrices with coefficients in $A$.  If $G$ is an abstract Lie
group, then it is an (abstract) \emph{matrix Lie group} when it admits an
embedding as a closed subgroup $G \into \M(d,\R)$ for some $d$; equivalently,
$G \into \M(d,\C)$ for some $d$.  (Such an embedding is never unique,
and the abstraction here is that we assume that an embedding exists
without choosing one.  Note also that $\GL(d,\R)$ is not closed in
$\M(d,\R)$, but it has a closed embedding in $\M(d+1,\R)$.)  In addition,
if $G$ has an embedding in some $\M(d,\R)$ in which it is defined by real
polynomial equations, then it is a \emph{real algebraic group}. Likewise,
$G$ is a \emph{complex algebraic group} if some embedding of it $\M(d,\C)$
is defined by complex polynomial equations.  If $G$ is real algebraic,
then it has a \emph{complexification} $G_\C$ given by taking the equations
that define $G$ in $\M(d,\R)$, and letting $G_\C$ be the complex solutions
to the same equations in $\M(d,\C)$.

Every compact Lie group is real algebraic, and every semisimple complex
Lie group is complex algebraic.  There is also a bijection between compact
semisimple and complex semisimple Lie groups via complexification in one
direction and the Cartan decomposition (see below) in the other direction.
Every semisimple real matrix Lie group is real algebraic (using any
matrix embedding) and thus has a complexification.  By contrast, many
semisimple real Lie groups are not matrix Lie groups, and thus do not
have complexifications.

Every Lie group $G$ has a Lie algebra $L = T_1(G)$ with an abstract bilinear
Lie bracket $[x,y]$.  There is also an abstract exponential map $\exp:L
\to G$ that generalizes the matrix exponential for matrix Lie groups,
and that yields the real analytic structure of $G$.  Given $L$, it is the
Lie algebra of a unique simply connected (or \emph{maximal}) Lie group
$G_{\max}$, which is a covering space of every other Lie group $G$ with
Lie algebra $L$. The fundamental group $\pi_1(G)$ (which is necessarily
abelian) lifts to a discrete subgroup of the center $Z(G_{\max})$, and
any discrete subgroup of $Z(\tG)$ is the fundamental group of some $G$
with Lie algebra $L$.

A Lie group $G$ is called \emph{simple} when its Lie algebra $L$ is simple,
\ie, $L$ does not possess a nontrivial normal Lie subalgebra (or ideal)
$N \subset L$.  Confusingly, $G$ is simple as a Lie group if and only
if it is quasisimple rather than simple as a group, where a group $G$ is
\emph{quasisimple} when it is a perfect central extension of a simple group.
A Lie group $G$ is simple as a group if and only if it is minimal simple
(see below) as a Lie group.

Given $g \in G$, the derivative at $h=1$ of the conjugation action $h
\mapsto h^g$ of $G$ on itself is the \emph{adjoint action} $x \mapsto x^g$
of $G$ acting on $L$.  This is a linear action and is thus also called
the \emph{adjoint representation}.

\begin{example} $\SL(2,\R)$ is a simple (real) algebraic group with
fundamental group $\Z$.  It is quasisimple as a group and its center
is $\Z/2$.  Its minimal quotient $\PSL(2,\R)$ is also an algebraic group
and also has fundamental group $\Z$.  These are the only two algebraic
groups with Lie algebra $\sl(2,\R)$, even though all covering spaces of
$\PSL(2,\R)$ have the same Lie algebra.  In particular, the universal or
infinite cyclic cover $\widetilde{\SL(2,\R)}$ has a different topology
from any finite cover of $\SL(2,\R)$.

In addition, $\SL(2,\R)$ and $\SU(2)$ have the same complexification
$\SL(2,\C)$, which is simply connected and thus does not have any nontrivial
covering spaces.  More generally, every semisimple complex Lie group $G_\C$
is the complexification of more than one semisimple real $G$.
\end{example}

A Lie group $G$ is \emph{semisimple} when its Lie algebra $L$ is semisimple,
meaning that $L$ is a direct sum of simple Lie algebras.  If $L$ or $G$
is semisimple, then $Z(G_{\max})$ is discrete, and $L$ also has a canonical
minimal Lie group
\[ G_{\min} \defeq G/Z(G) \]
using any Lie group $G$ with Lie algebra $L$.  The Lie group $G_{\min}$ is
algebraic, and it is also called an \emph{adjoint Lie group}, because it is
the image of $G$ in its adjoint action on $L$.  If $L$ is semisimple, then
it also has a maximal algebraic Lie group $G_{\amax}$, and $Z(G_{\amax})$
is finite.  If $L$ is semisimple and either compact (meaning that it has
a compact Lie group) or complex, then $G_{\amax} = G_{\max}$.

Suppose that
\[ L \cong L_1 \oplus L_2 \oplus \cdots \oplus L_q \]
is a factorization of a semisimple $L$ into simple summands.
If $G$ is any of $G_{\min}$, $G_{\amax}$, or $G_{\max}$, then $G$
factors the same way that $L$ does, \ie,
\begin{eq}{e:gprod} G \cong G_1 \times G_2 \times \cdots \times G_q, \end{eq}
where each factor $G_j$ is simple.  However, other forms of $G$ might
factor less than $L$ does, because $\pi_1(G)$ can be a diagonal subgroup
of $\pi_1(G_{\min})$.

\begin{example} The Lie algebra of $\SO(4)$ factors as
\[ \so(4) \cong \su(2) \oplus \su(2), \]
but $\SO(4)$ itself does not factor.  Instead, its simply connected double
cover is
\[ \Spin(4) \cong \SU(2) \times \SU(2), \]
and $\SO(4)$ double covers its minimal form
\[ \PSO(4) \cong \SO(3) \times \SO(3). \]
As indicated, both of these alternate forms factor.
\end{example}

Another criterion for $G$ or $L$ to be semisimple is that $L$ has a
nondegenerate \emph{Killing form} $F(x,y)$, which is a certain bilinear
form which is defined intrinsically from the Lie algebra structure of $L$.
(Explicitly, if $T_{ab}^c$ is the tensor of the Lie bracket $[x,y]$, then
$F$ is given by the tensor network formula $F_{ab} = T_{ac}^d T_{bd}^c$.)

If $G$ is semisimple, then it has a generalized polar decomposition which is
unique up to conjugation and which is called the \emph{Cartan decomposition}.
The Cartan decomposition is easier to describe at the Lie algebra level
first.  If $L$ is the Lie algebra of $G$, then its Cartan decomposition
is a splitting $L = L_K \oplus L_P$ such that: (1) The Killing form $F$
is positive-definite on $L_P$ and negative-definite on $L_K$; (2) $L_P$
and $L_K$ are Killing-orthogonal; and (3) $L_K$ is a Lie subalgebra.
The subalgebra $L_K$ exponentiates to an \emph{angular subgroup} $K
\subseteq G$, while the subspace $L_P$ exponentiates to a \emph{radial
manifold}  $P \subseteq G$.  Then each $g \in G$ factors uniquely as $g =
pk$ with $p \in P$ and $k \in K$.

The Cartan decomposition $G = PK$ lets us interpret $P$ as both a
submanifold of $G$, and as the quotient space $P \cong G/K$ of left cosets
of the subgroup $K$, with an associated quotient map $\pi_P:G \to P$.
In this second interpretation, $G$ acts transitively on $P$ and $P$ is a
(nonpositively curved) \emph{symmetric space}.  Note also that even when
$G = K$, which is the case if and only if $G$ is compact, it still has a
trivial Cartan decomposition in which $P$ is a single point.

\begin{example} If $G = \SL(d,\C)$, then up to conjugation $K = \SU(d)$,
and then $P$ is the manifold of positive-definite Hermitian matrices with
determinant 1.  In this case, the equation $g = pk$ is the determinant 1
special case of the standard left polar decomposition of a complex matrix,
with $p = \sqrt{gg^*}$.
\end{example}

If $G \subseteq \M(d,\C)$ is semisimple, then it lies in $\SL(d,\C)$;
moreover, up to conjugation, its angular subgroup $K_G \subseteq \SU(d)$
is unitary and thus compact.  In this position, the radial complement $P_G$
consists of those positive-definite Hermitian matrices that lie in $G$:
\[ P_G = P_{\SL(d,\C)} \cap G. \]

Every Lie group $G$ also has a maximal compact subgroup $C \subseteq G$
which is unique up to conjugation.  If $G$ is semisimple, then up to
conjugation, $C$ is related to the angular subgroup $K$ by the inclusions
\[ [K,K] \subseteq C \subseteq K. \]
If $G$ is algebraic, then necessarily $K = C$; otherwise, they are not
always equal.  Using the fact that $Z(G) \subseteq K$, another criterion
is that $K = C$ if and only if $G$ has finite center.  Regardless, the
commutator subgroup $[K,K]$ is always compact and semisimple.  It will be
useful to see the difference between $K$ and $[K,K]$ (and therefore the
options for $C$) using the abelianization
\[ K_{\ab} \defeq K/[K,K]. \]

If $G$ is simple, then either $K$ is semisimple, $K_{\ab}$ is trivial,
and $G$ and its Lie algebra $L$ are \emph{non-Hermitian type}; or
$K_{\ab}$ is 1-dimensional and $G$ and $L$ are \emph{Hermitian type}.
(The terminology comes from the fact that the symmetric space $P$ is a
$G$-invariant complex manifold with if and only if $G$ is Hermitian type.)
If $G$ is a general semisimple Lie group, then $L$ has some $m \ge 0$
Hermitian-type summands.  In this case, the abelianization of $K_{\min}
\subseteq G_{\min}$ is an $m$-dimensional torus:
\[ (K_{\min})_{\ab} \cong U(1)^m. \]
In a general $G$ with the same Lie algebra $L$, the abelianization $K_{\ab}$
can be any covering space of $(K_{\min})_{\ab}$.  In particular, the quotient
\[ E \defeq K/C \cong \R^j \]
is a $j$-dimensional abelian Lie group for some $j \le m$, corresponding
to the directions of $(K_{\min})_{\ab}$ that have been unrolled completely.

\begin{examples} $\SL(2,\R)$ is Hermitian type as implied above. In this
case, $K = \SO(2)$ is a circle and $[K,K]$ is trivial.  More generally,
the real symplectic group $\Sp(2d,\R)$ is always Hermitian type, with $K
\cong \U(d)$ and thus $K_{\ab} \cong U(1)$.  Thus, for any $d \ge 1$,
the Euclidean factor $E$ in $\widetilde{\Sp(2d,\R)}$ is 1-dimensional.

If $G = \widetilde{\SL(2,\R)}{}^2$, then its Euclidean factor $E$
is $2$-dimensional.  If instead we let $G$ be the quotient of
$\widetilde{\SL(2,\R)}{}^2$ by the diagonal subgroup of its center, then
$K_{\ab}$ is a cylinder and $E$ is 1-dimensional.
\end{examples}

\subsection{Non-matrix Lie groups}
\label{ss:nonmat}

If $G$ is a semisimple algebraic group, then it is straightforward to
describe $g \in G$ as a matrix using any particular matrix embedding $G
\into \M(d,\R)$.  If $G$ is semisimple and not algebraic, then it is not
a matrix Lie group and it is more complicated to define a data type to
describe $g \in G$.  In this subsection, we give a solution based on the
Cartan decomposition $G = PK$ and the covering space analysis of $K$
in \Sec{ss:lie}.

Let
\[ \rho:G \to \SL(d,\C) \]
be a matrix representation of $G$ with these two properties:
\begin{enumerate}
\item The kernel $\ker(\rho)$ is discrete; equivalently, $\rho$ is injective
on the Lie algebra $L$.  For example, we can let $\rho$ be the (complexified)
adjoint representation, so that $\rho(G) \cong G_{\min}$ and $\ker(\rho)
= Z(G)$.
\item $\rho$ has been adjusted by conjugation so that $\rho(K) \subset
\SU(d)$, the maximal compact subgroup of $\SL(d,\C)$.
\end{enumerate}
(Note that we could just as well use $\SL(d,\R)$ and $\SO(d)$ throughout.)
Given any such $\rho$, the Cartan decomposition in $\rho(G)$ is consistent
with the Cartan or polar decomposition in $\SL(d,\C)$, as explained
in \Sec{ss:lie}.  Moreover, $\rho$ is a diffeomorphism between $P$ and
$\rho(P)$, and a covering map from $K$ to its compact image $\rho(K)$.
\Sec{ss:lie} also gives us the explicit model $K \cong C \times E$, where
$C$ is compact and thus a matrix group, and $E$ is a (possibly trivial)
Euclidean factor.

Thus, we can explicitly describe $g \in G$ by a matrix $p \in \rho(P)$
and angular data $(c,v) \in C \times E$, with the interpretation that
\[ \rho(g) = p\rho(c,v). \]
In particular, this form can be used as the input to \Alg{a:l} below,
or any other algorithm promised by \Thm{th:semisimple}.

\begin{example} In the Lie group $G = \widetilde{\SL(2,\R)}$, the angular
subgroup is $K \cong \R$, while the maximal compact subgroup $C$ is trivial.
Let $\rho:G \to \SL(2,\R)$ be the natural representation.  Then we can
describe $g$ by a symmetric matrix $p \in \M(2,\R)$ and an angle $\theta
\in \R$ such that
\[ \rho(g) = p \begin{bmatrix} \cos(\theta) & -\sin(\theta) \\
    \sin(\theta) & \cos(\theta) \end{bmatrix}. \]
\end{example}

\subsection{Metrics on Lie groups}
\label{ss:metrics}

Given a semisimple Lie group $G$, we need an appropriate
metric $d_G(\cdot,\cdot)$ for both the statements and proofs of
Theorems~\ref{th:qudit} and \ref{th:semisimple}.  The main answer is
that $d_G(\cdot,\cdot)$ should be a left-invariant Finsler metric.
All left-invariant Finsler metrics on $G$ are bi-Lipschitz equivalent,
and thus equivalent in the statements of our theorems.  At the same time,
certain left-invariant Riemannian metrics, which are special cases of Finsler
metrics, are the most convenient for our purposes.  In this subsection,
we sketch these concepts in more detail.

If $X$ and $Y$ are two metric spaces, then a function $\tau:X \to Y$ is
\emph{Lipschitz} means that
\[ d_Y(\tau(p),\tau(q)) = O(d_X(p,q)) \]
for all $p,q \in X$ (nonasymptotically).  In case the specific
constant matters, $\tau$ is $c$-Lipschitz when
\[ d_Y(\tau(p),\tau(q)) \le c d_X(p,q) \]
for all $p,q \in X$.  If $\tau$ is a bijection, then it is
\emph{bi-Lipschitz} when both $\tau$ and $\tau^{-1}$ are Lipschitz,
equivalently that
\[ d_Y(\tau(p),\tau(q)) = \Theta(d_X(p,q)) \]
for all $p,q \in X$.  Two different metrics on the same space can likewise
be bi-Lipschitz equivalent by the same criterion.

A \emph{pseudo-metric} on a set $X$ is a function $d:X \times X \to \R_{\ge
0}$ that satisfies the axioms of a metric space, except that $d(p,q) =
0$ is allowed even when $p \ne q$.  The equation $d(p,q) = 0$ defines
an equivalence relation $p \sim q$, and $d$ then induces a metric on the
quotient $X/\sim$.  Conversely, if $\tau:X \to Y$ is a surjective function
and $Y$ has a metric, then it induces a pseudo-metric on $X$ by the formula
\[ d_X(p,q) \defeq d_Y(\tau(p),\tau(q)). \]

If $M$ is a smooth manifold, a \emph{Finsler metric} is a continuous
assignment of a norm $\norm{\cdot}_p$ on each tangent space $T_p(M)$.
Given a Finsler metric on $M$, we define the length of a smooth path
$\gamma:[a,b] \to M$ by the integral
\[ \len(\gamma) \defeq \int_a^b \norm{\gamma'(t)}_{p = \gamma(t)} dt, \]
which is a generalization of the standard integral for the length of a
parametric curve in Euclidean space.  Given $p,q \in M$ in a connected
Finsler manifold $M$, we define the metric $d_M(p,q)$ as the infimal length
of a path $\gamma$ from $p$ to $q$.  If there is a $\gamma$ that minimizes
this length, then it is called a \emph{geodesic}.  If $M$ is connected and
complete as a metric space (which is the case for most manifolds that we
consider in this article), then every two points $p,q \in M$ are connected
by at least one geodesic.

Given a differentiable map $f:M \to N$ between two Finsler manifolds,
each value $df(p)$ of the derivative $df$ is linear map
\[ df(p):T_p(M) \to T_{f(p)}(N), \]
which thus has an operator norm $\norm{df(p)}_\infty$.

If the norm $\norm{\cdot}_p$ on a Finsler manifold $M$ from a
positive-definite inner product $\braket{\cdot,\cdot}_p$, then the metric
is \emph{Riemannian}.  If $M$ and $N$ are two Riemannian manifolds,
then there is a standard Riemannian metric on the product given by
taking the sum of the inner products on $M$ and $N$:
\begin{gather*}
\braket{(x,y),(x,y)}_{(p,q)} \defeq \braket{x,x}_p + \braket{y,y}_q \\
    p \in M \quad q \in N \quad x \in T_p(M) \quad y \in T_q(N).
\end{gather*}
This product metric then satisfies the Pythagorean theorem:
\begin{eq}{e:pythag} d_{M \times N}((p_1,q_1),(p_2,q_2))
    = \sqrt{d_M(p_1,p_2)^2 + d_N(q_1,q_2)^2}. \end{eq}

If a group $G$ acts on a smooth manifold $M$, then a Finsler metric
$\norm{\cdot}_p$ or a Riemannian inner product $\braket{\cdot,\cdot}_p$ on
$M$ may or may not be $G$-invariant.  In particular, if $G = M$ is a Lie
group, then any such structure may be left- or right-invariant (if it is
invariant under left or right multiplication by $G$), or bi-invariant (if
it is both).   If $G$ is any Lie group, then there is a bijection between
left-invariant Finsler metrics on $G$ and Banach norms $\norm{\cdot}_L$
on the Lie algebra $L = T_1(G)$:  If $x \in T_g(G)$ for some other $g \in
G$, then left multiplication gives us $g^{-1}x \in L$, and we define
\[ \norm{x}_g \defeq \norm{g^{-1}x}_L. \]
Likewise, there is a bijection between left-invariant Riemannian metrics
on $G$ and inner products $\braket{\cdot,\cdot}_L$ on $L$.  On the other
hand, if $G$ is semisimple, then it has a bi-invariant Finsler metric if
and only if it is compact.  If $G$ is compact, then it has a bi-invariant
Riemannian metric (see below).

To keep track of numerical precision, we need a counterpart of \equ{e:bbmul}
when $G$ is not compact and has a left-invariant Finsler metric.  In this
case
\[ B_G(\eps,1)h \subset hB_G(N\eps,1) = B_G(N\eps,h) \qquad N = O(d(h,1)), \]
which implies that:
\begin{eq}{e:bmul} \begin{gathered}
B_G(\eps,g)B_G(\delta,h) \subseteq B_G(N\eps+\delta) \\
    B_G(\delta,h^{-1}) \subseteq B_G(N\delta,h)^{-1}.
\end{gathered} \end{eq}

Since the Lie algebra $L$ is a finite-dimensional vector space, all norms on
it are bi-Lipschitz equivalent.   Thus, all left-invariant Finsler metrics on
a Lie group $G$ are also bi-Lipschitz equivalent.  Given this equivalence,
we adopt a specific left-invariant Riemannian metric on any semisimple $G$
as follows:

As in \Sec{ss:nonmat}, let $\rho:G \to \SL(d,\C)$ be a matrix representation
which is faithful on $L$ and unitary on $K$.  Then $L$ inherits a
Hilbert--Schmidt norm and inner product via $\rho$:
\begin{eq}{e:hsnorm} \braket{x,y}_L \defeq \Re(\tr(\rho(x)\rho(y)^*))
    \qquad  \norm{x}_L \defeq \sqrt{\braket{x,x}_L}. \end{eq}
This norm is also $K$-adjoint invariant, \ie, invariant under conjugation
by the angular subgroup $K$.  It therefore extends to a left-invariant
Riemannian metric $d_G(g,h)$ on $G$ which is also $K$-right-invariant.
If $G = K$ is compact, then this metric is bi-invariant.  When $G$
is semisimple, we can fully standardize the metric \eqref{e:hsnorm}
by letting $\rho$ be the adjoint representation, although this will not
materially matter in this article.  Finally, since the metric $d_G(g,h)$
is right $K$-invariant, any two left cosets $gK$ and $hK$ are at a constant
distance from each other.  It follows that $d_G$ induces a $G$-invariant
Riemannian metric $d_P$ on the symmetric space $P \cong G/K$.

It can be difficult to accurately compute distances $d_M(p,q)$ in a general
Riemannian or Finsler manifold $M$, because solving for geodesics carries
the complications of solving ODEs as well as combinatorial shortest-path
algorithms.  For our purposes, we will use one important case where there is
an exact algorithm, and two different types of approximation.

\begin{proposition} If $G$ is a semisimple Lie group with representation
$\rho$ as above, then distances in the symmetric space $P \cong G/K$
are given by
\[ d_P(\pi_P(g),\pi_P(h)) =
    \frac{\norm{\log(\rho(h^{-1}g)\rho(h^{-1}g)^*)}_2}2, \]
where $\log$ is the unique Hermitian logarithm of its positive Hermitian
argument, and
\[ \norm{x}_2 \defeq \sqrt{\tr(xx^*)} \]
is the Hilbert--Schmidt norm on $M(d,\C)$.
\label{p:pdist} \end{proposition}

\begin{theorem} The multiplication map
\[ \mu:P \times K \to G \]
is a bi-Lipschitz diffeomorphism, where $P \times K$ is given the
product metric $d_{P \times K}$.
\label{th:pkbilip} \end{theorem}

We will establish \Thm{th:pkbilip} as a corollary of \Thm{th:bilip}
in \Sec{s:bilip}.

The second approximation that we will need is
\begin{eq}{e:gnest} d_G(g,1) = \norm{g-1}_2(1+o(1)) \end{eq}
in the limit as $g \to 1$.  The right side of \equ{e:gnest} only makes
sense when $G$ is a matrix Lie group, but we can use it for arbitrary $G$
by replacing $G$ by $G_{\min}$, provided that $G$ is in an open region
around $1 \in G$ which is also an open region in $G_{\min}$.  The equation
is then an elementary estimate in Riemannian geometry that amounts to the
fact that the metric inner product $\braket{\cdot,\cdot}_G$ matches the
Hilbert--Schmidt inner product on the tangent space $T_1(G)$.

\begin{proof}[Proof of \Prop{p:pdist}] We can replace $G$ by $\rho(G)$ in
the question, since they have compatible Cartan decompositions.  Thus we
can assume that $\rho$ is the identity.

We use the fact that the exponential map $\exp:L_P \to P$ is a
diffeomorphism between the radial manifold $P$ and its Lie tangent space
$L_P$.  Moreover, exponential curves from the identity are also geodesics
in both $G$ and $P$, so that for any $x \in L_P$,
\[ d_G(\exp(x),1) = d_P(\exp(x),1) = \norm{x}_L = \norm{x}_2. \]
In context, $x \in L_P$ is a Hermitian matrix and $p = \exp(x)$ is positive
Hermitian.  Thus we can also write this equation as
\[ d_P(p,\pi_P(1)) = \norm{\log(p)}_2. \]
Using polar decomposition $g = pk$, it follows that
\[ d_P(\pi_P(g),1) = \norm{\log(\sqrt{gg^*)}_2}
    = \frac{\norm{\log(gg^*)}_2}2. \]
Finally the $G$-invariance of $d_P$ yields
\[ d_P(\pi_P(g),\pi_P(h)) = \frac{\norm{\log((h^{-1}g)(h^{-1}g)^*)}_2}2, \]
as desired.
\end{proof}

Recall that $G$ has a maximal compact subgroup $C \subseteq K \subseteq G$,
which yields a quotient map
\[ \pi_Q:G \to G/C \defeq Q. \]
Any two left cosets $gC$ and $hC$ are also at a constant distance from
each other, so that $d_G(g,h)$ also induces a $G$-invariant Riemannian
metric $d_Q$ on $Q \defeq G/C$.  The quotient map $\pi_Q$ is distance
nonincreasing (or 1-Lipschitz), but cannot decrease distances by more
than the diameter of $C$, which is finite.  In other words,
\begin{eq}{e:gqdist} d_G(g,h) = d_Q(\pi_Q(g),\pi_Q(h)) + O_+(1). \end{eq}
If $K = C$, then $\pi_P = \pi_Q$, and \Prop{p:pdist} along with
\eqref{e:gqdist} make $P$ a useful simplification of the metric geometry
of $G$.   If $K \ne C$ is not compact, then we further simplify the
geometry of $Q$.   Following \Sec{ss:lie}, the quotient $E = K/C \cong
\R^j$ becomes a Euclidean space under its left-invariant metric $d_E$.
(There is also an isometric group isomorphism $K \cong E \times C$,
but we will not need this in this article.)

Since the metrics $d_{P \times K}$ and $d_G$ are both right $C$-invariant,
we also obtain the following corollary of \Thm{th:pkbilip}:

\begin{corollary} The multiplication map $\mu$ descends to
a bi-Lipschitz diffeomorphism
\[ \mu_Q:P \times E \to Q \qquad P \times E \cong (P \times K)/C
    \qquad Q \cong G/C \]
given by quotienting by right multiplication by $C$.
\label{c:pebilip} \end{corollary}

\Cor{c:pebilip} and \equ{e:gqdist} together yield
\[ d_G(g,h) = \Theta(d_{P \times E}(\pi_Q(g),\pi_Q(h))) \]
as $d_G(g,h) \to \infty$.  In other words, $\mu_Q^{-1} \circ \pi_Q$ is
a \emph{quasiisometry}.

We conclude this subsection with a proof of \Cor{c:matnorm}.

\begin{proof}[Proof of \Cor{c:matnorm}] Under the hypotheses
of \Thm{th:semisimple} and the corollary, it suffices to show that
\[ \norm{g-w}_\infty = O(\norm{g}_\infty d_G(g,w)) \]
uniformly in $g$ as $d_G(g,w) \to 0$, and that
\[ d_G(g,1) =  O(\log(\norm{g}_\infty)+1) \]
for all $g$ (nonasymptotically).

For the first claim, \equ{e:gnest} tells us that
\[ d_G(g,w) = \norm{1 - g^{-1}w}_2(1+o(1)). \]
as $d_G(g,w) \to 0$.  Meanwhile
\[ \norm{g-w}_\infty \le \norm{g}_\infty \norm{1-g^{-1}w}_\infty
    = \norm{g}_\infty O(\norm{1-g^{-1}w}_2). \]
The claim follows by combining these estimates.

The second claim follows from \Prop{p:pdist} (or a step in its proof)
together with \eqref{e:gqdist}:
\[ d_G(g,1) = O(\norm{\log(gg^*)}_\infty+1)
    = O(\log(\norm{g}_\infty)+1). \qedhere \]
\end{proof}

\subsection{Cancellation and nilpotence}
\label{ss:cannil}

In this subsection, we will discuss results of Elkasapy and Thom
\cite{ET:lower,Elkasapy:lower} concerning the nilpotence degree
$\nil(\omega)$ of a group word $\omega$, and its relation to the conjugate
cancellation degree $\ccan_G(\omega)$ that we need.

Recall that every group $G$ has a \emph{lower central series}
\[ G = G_1 \supseteq G_2 \supseteq G_3 \supseteq \cdots \]
which is defined recursively by the formula
\[ G_{n+1} \defeq \comm{G_n,G}. \]
An induction argument using the Hall--Witt identity
\[ \bigcomm{\comm{g,f},h^f} \bigcomm{\comm{f,h},g^h}
    \bigcomm{\comm{h,g},f^g}=1 \]
then tells us that
\[ \comm{G_n,G_k} \subseteq G_{n+k} \]
for all $n$ and $k$.  If $G_{n+1} = 1$ for some $n$, then $G$ is called
\emph{$n$-step nilpotent}.  In particular, if $G = F_k$ is a free
group with $k$ generators and $\omega \in F_k$ is a group word, then the
\emph{nilpotence degree} $\nil(\omega)$ is the largest $n$ such that $\omega
\in (F_k)_n$.  Every nontrivial word $\omega$ has finite nilpotence degree
(which is part of \Thm{th:cannil} below, or can be proven in other ways).
If $\nil(\omega) = n$, then its evaluation $\omega(g_1,g_2,\ldots,g_k)$
is trivial in any $(n-1)$-step nilpotent group $G$, but it is nontrivial
in the \emph{free $n$-step nilpotent group}
\[ N_{k,n} \defeq F_k/(F_k)_{n+1}. \]

Elkasapy and Thom \cite{ET:lower} considered (in the case $k=2$) the minimum
word length to reach any given nilpotence degree, as in this definition:
\[ \ell_{k,n} \defeq \min_{\omega \in (F_k)_n} \len(\omega). \]
They also considered the corresponding exponent
\[ \lambda_{k,n} \defeq \frac{\log(\ell_{k,n})}{\log(n)}, \]
and they showed (Lemma 2.1) that
\[ \lambda \defeq \inf_n \lambda_{2,n}
    = \lim_{n \to \infty} \lambda_{2,n}. \]
By embedding $F_k$ into $F_2$, we also learn that $\ell_{k,n} =
\Theta(\ell_{2,n})$ for any fixed $k$, as $n \to \infty$.  This implies that
\[ \lambda = \lim_{n \to \infty} \lambda_{k,n} \]
for any fixed $k$.  With a little more work (which we will omit here),
one can also show that
\[ \lambda = \inf_n \lambda_{k,n} \]
for each $k$.   Elkasapy and Thom also explain that lower bounds due to
Fox and Malestein--Putman imply that $\lambda \ge 1$.

The group commutator $\comm{g,h}$ shows that $\lambda \le 2$, but, starting
with \eqref{e:len14}, Elkasapy and Thom realized that $\lambda < 2$.

\begin{example} Refining \eqref{e:len14}, one of the words found by
Elkasapy and Thom is:
\[ \omega(g,h) = \bigcomm{\comm{g,h},\comm{h,g^{-1}}}
    = ghg^{-2}h^{-1}ghgh^{-1}g^{-2}hgh^{-1}. \]
This $\omega$ also has length 14, but now $\nil(\omega) = 5$.
\end{example}

In a followup paper \cite{Elkasapy:lower}, Elkasapy found a sequence
of words $\omega_n \in F_2$ that imply that $\lambda \le \log_\phi(2)$.
In \Sec{s:elkasapy}, we will state his precise result as \Thm{th:elkasapy},
and give his proof with minor changes.

We establish the following comparisons between $\nil(\omega)$, and
the statistics $\can_G(\omega)$ and $\ccan_G(\omega)$ as defined in
\equ{e:candeg} and the surrounding discussion in \Sec{ss:steps}.

\begin{theorem} For any nontrivial word $\omega \in F_k$ and any semisimple
Lie group $G$,
\begin{eq}{e:cnle} \nil(\omega) \le \can_G(\omega)
    \le \ccan_G(\omega) < \infty. \end{eq}
Conversely,
\begin{eq}{e:cnlim} \lim_{d \to \infty} \ccan_{\SU(d)}{\omega}
    = \nil(\omega). \end{eq}
\label{th:cannil} \end{theorem} \eatline

We introduce some definitions and facts that we will use in the proof
of \Thm{th:cannil} and afterwards.  If $G$ is an algebraic group, we
define the group $G[[\eps]]$ to be the group of $\R[[\eps]]$-points of
$G$, where $\R[[\eps]]$ is the ring of formal power series in $\eps$.
More explicitly, if $G$ is the set of solutions in $\M(d,\R)$ to a set of
polynomial equations, then $G[[\eps]]$ is by definition the set of solutions
in $\M(d,\R[[\eps]])$ to the same equations.  We define $G[[\eps]]_0$ to be
the subgroup of $G[[\eps]]$ of those elements that become the identity $1 \in
G$ if we set $\eps = 0$. The group $G[[\eps]]$ is then a semidirect product
\begin{eq}{e:semidir} G[[\eps]] \cong G \ltimes G[[\eps]]_0. \end{eq}
We will also use the quotient group $G[[\eps]]/\eps^n$, by definition
the set of  $\R[[\eps]]/\eps^n$-points of $G$, and the corresponding
subquotient $G[[\eps]]_0/\eps^n$.

Although $G[[\eps]]$ and $G[[\eps]]_0$ are not strictly Lie groups (because
they are infinite-dimensional over $\R$), they have associated Lie algebras
\[ L[[\eps]] \defeq L \tensor_\R \R[[\eps]]
    \qquad L[[\eps]]_0 \defeq \eps L[[\eps]]. \]
We collect two basic facts about $G[[\eps]]_0$ and $L[[\eps]]_0$
as a lemma.

\begin{lemma} Given $G[[\eps]]_0$ and $L[[\eps]]_0$ defined as above:
\begin{enumerate}
\item The exponential map is a bijection and thus has an inverse:
\[ \exp:L[[\eps]]_0 \to G[[\eps]]_0 \qquad
    \log:G[[\eps]]_0 \to L[[\eps]]_0. \]
\item The BCH formula \eqref{e:bch} converges on $L[[\eps]]_0$ and matches
the group law on $G[[\eps]]_0$.
\item Given a word $\omega \in F_k$, $\can_G(\omega)$ as defined by
\eqref{e:candeg} equals the largest integer $n$ such that
\begin{eq}{e:can} \omega(g_1,\ldots,g_k) = 1 \in G[[\eps]]_0/\eps^n \end{eq}
for all inputs in $G[[\eps]]_0/\eps^n$, and $\ccan_G(\omega)$ equals the
largest $n$ such that \eqref{e:can} holds for inputs in $G[[\eps]]_0/\eps^n$
that are conjugate in $G[[\eps]]/\eps^n$.
\end{enumerate}
\label{l:gleps0} \end{lemma}

If $G$ is any Lie group with Lie algebra $L$, then we can use the
conclusions of \Lem{l:gleps0} as a definition of $G[[\eps]]$.  We can define
$G[[\eps]]_0$ to be $L[[\eps]]_0$ with the BCH product as its group law,
and then use \eqref{e:semidir} to define $G[[\eps]]$ as the semidirect
product of $G$ and $G[[\eps]]_0$.

\begin{proof} Part 1 is the statement that
for any $g \in G[[\eps]]/\eps^{n+1}$,
\begin{eq}{e:gexp} g = \exp(\eps x_1 + \eps^2 x_2 + \cdots
    + \eps^{n-1} x_n) \end{eq}
always has a unique solution for each $x_j \in L$.  This can be checked
by induction on $n$:  Given $x_1,\ldots,x_{n-1}$, the equation for $x_n$
becomes linear and nonsingular.

In part 2, the BCH formula converges because the factors of $\eps$
in its inputs accumulate.  The fact that it matches the group law in
$G[[\eps]]_0/\eps^n$ is the same calculation as the usual interpretation
of the BCH formula as a convergent Taylor series expansion.

In part 3, \equ{e:can} seems more general than \equ{e:candeg}, since in the
former an input $g$ has the form \eqref{e:gexp}, while in the latter $g =
\exp(\eps x)$.  This extra generality is superficial, since we can obtain
\equ{e:can} from \equ{e:candeg} by taking partial derivatives.
\end{proof}

Since the construction of $G[[\eps]]_0$ via the BCH formula only depends
on the Lie algebra $L$, the statistic $\can_G(\omega)$ only depends on $L$
and not otherwise on $G$.  Likewise $\ccan_G(\omega)$ only depends on $L$
and the adjoint action of $G$ on $L$.  Thus
\[ \ccan_G(\omega) = \ccan_{G_{\min}}(\omega) \]
when $G$ is semisimple, and we can assume at this moment that $G = G_{\min}$.
If $G = G_{\min}$, then it is a product of simple factors by \equ{e:gprod}.
Then $\can_G(\omega)$ and $\ccan_G(\omega)$ are each the minimum of their
values for the separate factors:
\begin{align*}
\can_G(\omega) &= \min_j \can_{G_j}(\omega) \\
\ccan_G(\omega) &= \min_j \ccan_{G_j}(\omega).
\end{align*}

\begin{proof}[Proof of \Thm{th:cannil}] The group $G[[\eps]]_0/\eps^n$ is
at most $(n-1)$-step nilpotent, since the factors of $\eps$ accumulate when
we take group commutators.  (It is exactly $(n-1)$-step nilpotent when $G$
is semisimple, although we will not need this refinement here.)  Taking $n =
\nil(\omega)$, the universal property of the free nilpotent group $N_{k,n-1}$
tells us that $\nil(\omega) \le \can_G(\omega)$, as desired.  Meanwhile,
$\can_G(\omega) \le \ccan_G(\omega)$ straight from the definition.  Indeed,
these two inequalities in \eqref{e:cnle} hold in any Lie group $G$.

The final inequality in \eqref{e:cnle}, that $\ccan_G(\omega) < \infty$
whenever $\omega$ is a nontrivial group word, is not true for solvable
Lie groups.  To prove it when $G$ is semisimple, we first consider
the special case $G = \PSL(2,\C)$.  It was first proven by Schottky
\cite{Schottky:mehrfach} that $\PSL(2,\C)$ contains nonabelian free groups.
In particular, if $a,b \in \PSL(2,\C)$ generate a free group $F_2$, then
the elements $b^nab^{-n}$, for all integers $n \in \Z$, are all conjugate
and are all freely independent as well; they generate an $F_\infty$.
It follows that if $\omega \in F_k$ is any nontrivial word, then it
is a nonconstant function on conjugate group elements in $\PSL(2,\C)$.
It is also a complex analytic function (indeed, a polynomial), so it has
a nontrivial power series at the identity, which means that
\[ \ccan_{\PSL(2,\C)}(\omega) < \infty. \]

We can prove that $\ccan_G(\omega) < \infty$ for any semisimple real Lie
group $G$ by stepping through the following reductions:
\begin{enumerate}
\item Since $\ccan_G(\omega)$ depends only on the Lie algebra $L$ of $G$
and the action of $G_{\min}$ on $L$, we can replace $G$ by $G_{\min}$ and
assume that $G$ is algebraic.
\item Since \equ{e:can} is an algebraic identity, we can replace a semisimple
real algebraic $G$ by its complexification $G_\C$.
\item If $G$ is complex semisimple, then by the classification of such Lie
groups, $G$ contains either $\SL(2,\C)$ or $\PSL(2,\C)$ as a subgroup. In
the former case, we can again pass to $\SL(2,\C)_{\min} \cong \PSL(2,\C)$.
\end{enumerate}

\Equ{e:cnlim} requires a different approach.  We can replace $\SU(d)$
by its complexification $\SL(d,\C)$, since (as in the previous paragraph)
\[ \ccan_{\SU(d)}(\omega) = \ccan_{\SL(d,\C)}(\omega). \]
It suffices to show that the free nilpotent group $N_{k,n-1}$ embeds in
$\SL(d,\C)[[\eps]]_0/\eps^n$ when $d$ is large enough.  We can make such
an embedding using a free nilpotent algebra, although more economical
embeddings also exist.

For any commutative ring $A$, let $E_{k,n-1}(A)$ be the $(n-1)$-step
\emph{free nilpotent algebra} of noncommutative polynomials in
$x_1,\ldots,x_k$ with coefficients in $A$, and with the relations that
all monomials of degree $n$ vanish.  We can identify
\[ E_{k,n-1}(A) \cong A^d \qquad d = d(k,n) = \frac{k^n-1}{k-1} \]
by using monomials of degree at most $n-1$ in shortlex order as a free
basis of $E_{k,n-1}(A)$.  Also, the given value of $d$ is the number of
nonzero monomials.  Since $E_{k,n-1}(A)$ acts on $E_{k,n-1}(A) \cong A^d$
by left multiplication, it has a faithful representation
\[ \rho_A:E_{k,n-1}(A) \into \M(d,A) \]
called the \emph{regular representation}.  Using the given ordering on
monomials, every nonconstant monomial $w \in E_{k,n-1}(A)$ is represented
by a strictly upper triangular matrix $\rho_A(w)$ whose entries are all 0
or 1.  We are interested in the special case $A = \C[[\eps]]/\eps^n$:
\begin{eq}{e:epsreg} \begin{gathered}
\rho_{\C[[\eps]]/\eps^n}:E_{k,n-1}(\C[[\eps]]/\eps^n)
    \into \M(d,\C[[\eps]]/\eps^n) \\
    \rho_\eps \defeq \rho_{\C[[\eps]]/\eps^n}.
\end{gathered} \end{eq}

It is a standard fact \cite[Sec.~6]{Khukhro:formula} \cite[Ch.~5]{MKS:groups}
that the free nilpotent group $N_{k,n-1}$ with generators $g_1,\ldots,g_k$
has an embedding
\[ \sigma:N_{k,n-1} \into E_{k,n-1}(\C) \qquad \sigma(g_j) = \exp(x_j). \]
We can modify this embedding as
\[ \sigma_\eps:N_{k,n-1} \into E_{k,n-1}(\C[[\eps]]/\eps^n)
    \qquad \sigma_\eps(g_j) = \exp(\eps x_j). \]
Composing with the regular representation $\rho_\eps$, each generator
$\rho_\eps(\sigma_\eps(g_j))$ is a unitriangular matrix such that $\eps$
divides all entries above the diagonal.  Thus the composition is an
injective representation
\[ \rho_\eps \circ \sigma_\eps:N_{k,n-1} \into \SL(d,\C)[[\eps]]_0/\eps^n. \]

We also claim that the group generators $\rho_\eps(\sigma((g_j))$ are all
conjugate in $\SL(d,\C)[[\eps]]/\eps^n$.  To see this, we consider the
action of $\SL(k,\C)$ on $E_{k,n-1}(\C[[\eps]]/\eps^n)$ given by applying
the linear maps in $\SL(k,\C)$ to the generators $x_1,\ldots,x_k$ and then
extending them to algebra automorphisms.  Again using \eqref{e:epsreg},
this action yields an embedding
\[ \tau:\SL(k,\C) \into \SL(d,\C)[[\eps]]/\eps^n. \]
Now choose $f \in \SL(k,\C)$ such that $f(x_1) = x_j$ for some $j$.
Since $\tau$ acts by algebra automorphisms, it follows that
\[ \rho_\eps(x_1)^{\tau(f)} = \rho_\eps(x_j) \qquad
    \rho_\eps(\sigma_\eps(g_1))^{\tau(f)} = \rho_\eps(\sigma_\eps(g_j)). \]
Thus the matrices $\rho(\sigma(g_j))$ are all conjugate in
$\SL(d,\C)[[\eps]]/\eps^n$, which tells us that
\[ \nil(\omega) \le n \quad \implies \quad
    \ccan_{\SL(d,\C)}(\omega) \le n \]
for the given value of $d$, depending on $n$ and $k$.  Since
$\ccan_{\SL(d,\C)}(\omega)$ is monotonic in $d$, the limit statement
\eqref{e:cnlim} follows.
\end{proof}

A step beyond \equ{e:can}, we can also consider the first nontrivial term
$\omega^{(n)}$ in the power series expansion of $\omega$ as a function
on conjugate elements of $L$.  More precisely, by \Lem{l:gleps0},
we can refashion $\omega$ as a map
\[ \omega_{\log}:L[[\eps]]_0^k \to L[[\eps]]_0. \]
If $\ccan_G(\omega) = n$, then
\begin{multline*}
\omega_{\log}(\eps x^{u_1},\eps x^{u_2},\ldots,\eps x^{u_k}) \\
    = 1 + \eps^n \omega^{(n)}(x^{u_1},x^{u_2},\ldots,x^{u_k})
    \in L[[\eps]]_0/\eps^{n+1}.
\end{multline*}
here $\omega^{(n)}$ is a partial map from $L^n$ to $L$ which is defined
for conjugate inputs, and such that
\begin{eq}{e:ucc} \omega^{(n)}(x^{u_1},x^{u_2},\ldots,x^{u_k}) \ne 0 \end{eq}
for some $x \in L$ and some $u_1,u_2,\ldots,u_k \in G$.  More sharply,
if \eqref{e:ucc} holds for every $x \in L \setminus \{0\}$ and for some
$u_1,u_2,\ldots,u_k \in G$, then we say that $\omega$ has \emph{uniform
conjugate cancellation} in $G$, or for short that $\omega$ is UCC in $G$.

\Alg{a:sb} implicitly uses the fact that every $\omega$ is UCC in $\SU(2)$,
which holds because all nonzero adjoint orbits in $\su(2)$ are proportional.

Note also that if $G$ factors according to \equ{e:gprod},
then $\omega$ is UCC in $G$ if and only if:
\begin{enumerate}
\item $\omega$ is UCC in each $G_j$, and
\item $\ccan_{G_j}(\omega)$ is equal for all $j$.
\end{enumerate}

The following lemma yields a simpler algorithm for \Thm{th:semisimple}
when $G$ is compact; in particular, a simpler algorithm for \Thm{th:qudit}
than in the most general case.

\begin{lemma} If $G$ is a compact, semisimple Lie group, then each Elkasapy
word $\omega_n$ as defined in \Sec{s:elkasapy} is UCC in $G$.
\label{l:elkucc} \end{lemma}

\begin{remark} When $n \ge 5$, $\omega_n$ is not UCC in $\SL(2,\R)$.
\end{remark}

\begin{proof} \Thm{th:elkasapy} and \Lem{l:nilfib} tell us that
\[ \ccan_G(\omega_n) = f_n, \]
the $n$th Fibonacci number.  Given $z \in L \setminus \{0\}$, we will find
$u,v \in G$ such that
\[ \omega^{(f_n)}_n(z^u,z^{vu}) \ne 0. \]
We will use the complexifications $G_\C$ and $L_\C$ of $G$ and $L$, where
(among other things) we can multiply elements of $L$ by $i$.

The argument uses the proof of \Lem{l:nilfib}.  Following that
proof, let
\begin{align*}
\omega_1(g,h) = g &= \exp(\eps iz_1 + O(\eps^2)) \\
\omega_2(g,h) = h &= \exp(\eps iz_2 + O(\eps^2)),
\end{align*}
where
\[ z_1,z_2 \in iL \qquad z_1 = -iz^u \qquad z_2 = -iz^{vu} = z_1^v. \]
The recurrence \eqref{e:elk2} and \equ{e:glbracket} tell us that
\begin{eq}{e:zelk} \begin{aligned}
\omega_n(g,h) &= \exp(i\eps^{f_n} z_n + O(\eps^{f_n+1}))\\
    z_{n+2} &= -i[z_{n+1},z_n]
\end{aligned} \end{eq}
for $n \ge 3$.  We will show that if $g, h \in J_\eps$ are suitably chosen,
then each $z_n$ is nonzero and thus that each $\omega_n$ has uniform
conjugate cancellation.

We use the Serre relations \cite[Sec.~4.8]{Varadarajan:gtm} for $L_\C$.
The Lie algebra $L_\C$ has some (complex) rank $r > 0$ and a Cartan matrix
$A$, which is a nonsingular integer $r \times r$ matrix such that $A_{j,k}
\ge 0$ and $A_{j,j} = 2$ for all $j$ and $k$.  Then $L_\C$ is generated
by elements $h_j$, $x_j$ and $y_j$ with $1 \le j \le r$, subject to the
relations
\begin{align*}
[h_j,h_k] &= 0 & [x_j,y_j] &= h_j \\
[h_j,x_k] &= A_{j,k} x_k & [h_j,y_k] &= -A_{j,k} y_k,
\end{align*}
and some other relations that we will not need.  We can define an abstract
conjugate-linear involution $*$ on $L_\C$ (generalizing the Hermitian
adjoint of matrices) by the formulas
\[ h_j^* = h_j \qquad x_j^* = y_j \qquad y_j^* = x_j. \]
The Lie algebra $L$ is then the anti-self-adjoint subspace of $L_\C$.

Since $z \in L \setminus \{0\}$, and since $G$ is compact, $z$ is a
diagonalizable element of $L$ with an imaginary spectrum.  Thus, we can
find $u \in G$ to place $z_1 = -iz^u$ in the real linear span of the
elements $h_j$.  Since $A$ is nonsingular and $z$ is nonzero, there is
a $j$ such that
\[ [z_1,x_j] = a x_j \qquad [z_1,y_j] = -a y_j \qquad a \ne 0. \]
Without loss of generality, we can assume that $j=1$, since we can permute
the Serre generators; and $a=1$, since we can rescale $z$ and $z_1$
in tandem.  We can thus write
\[ z_1 = \frac{h_1}2 + z_0 \qquad [z_0,x_1] = [z_0,y_1] = 0. \]
The elements $x_1$, $y_1$, and $h_1$ span a Lie algebra isomorphic to
$\sl(2,\C)$ in $L_\C$.   We explicitly describe such an isomorphism with
a Lie algebra representation $\rho$ of these three elements of $L_\C$:
\begin{align*}
\rho(x_1) &= \begin{bmatrix} 0 & 1 \\ 0 & 0 \end{bmatrix} = \frac{X+iY}2 \\
\rho(y_1) &= \begin{bmatrix} 0 & 0 \\ 1 & 0 \end{bmatrix} = \frac{X-iY}2 \\
\rho(h_1) &= \begin{bmatrix} 1 & 0 \\ 0 & -1 \end{bmatrix} = Z
\end{align*}
Since $z_0$ is linearly independent from $x_1$, $y_1$, and $h_1$ and commutes
with them, we can extend $\rho$ to $z_0$ by setting $\rho(z_0) = 0$.  We set
\[ v = \exp\bigl(\frac{\pi i}4(x_1+y_1)\bigr), \]
and we can say both that $v \in G$ and that $\rho(v) \in \PSU(2)$ by
exponentiating matrices.  We obtain
\begin{align*}
\rho(v) =  \frac{1}{\sqrt{2}} \begin{bmatrix} 1 & i
    \\ i & 1 \end{bmatrix} \qquad
\rho(h_1^v) = \begin{bmatrix} 0 & -i \\ i & 0 \end{bmatrix}
    = Y.
\end{align*}
If we let $\omega_2 = \omega_1^v$, then
\begin{align*}
z_1 &= \frac{h_1}2 + z_0 & \rho(z_1) = Z \\
z_2 &= z_1^v = \frac{i(y_1-x_1)}2 + z_0 & \rho(z_2) = Y
\end{align*}
Using \equ{e:zelk} and induction, we obtain that
\[ \rho(z_n) = \begin{cases}
Z/2 & n \equiv 1 \pmod 3 \\
Y/2 & n \equiv 2 \pmod 3 \\
X/2 & n \equiv 0 \pmod 3
\end{cases} \]
for all $n$, and the $z_0$ term disappears from $z_n$ when $n \ge 3$.
We thus learn that
\[ \omega^{(f_n)}(z^u,z^{vu}) = iz_n \ne 0 \]
for all $n$, as desired.
\end{proof}

\section{Proof of \Thm{th:semisimple}}
\label{s:proof}

In this section, we generalize and extend Algorithms~\ref{a:sb}
and \ref{a:zb} to any semisimple Lie group $G$, in order to prove
\Thm{th:semisimple}.  The generalization must address these four
complications:
\begin{enumerate}
\item Especially if $G$ is not compact, then we need an initial stage of
long-distance golf to get within a bounded distance of a target element
$g \in G$.
\item We need more complicated zigzags when $G$ is any simple Lie group
other than $\SU(2)$ or $\SO(3)$, because no individual conjugacy class
points in all directions.
\item The conjugate cancellation degree of a word $\omega$ might not be
constant on different conjugacy classes, which complicates the construction
of roughly exponential steps.
\item If $G$ is semisimple but not simple, then a word $\omega$ could
have different cancellation degrees in different simple factors, which
also complicates the roughly exponential steps.
\end{enumerate}

\subsection{Long-distance golf}
\label{ss:long}

Most of the methods in this subsection are important when the Lie group
$G$ is not compact.  If so, then the distance $d_G(g,1)$ of a target element
$g \in G$ can be arbitrarily large, while the gate set $A$ is finite
and is thus a bounded distance from 1.  Thus if $d_G(w,g) = O(1)$ for an
$A$-word $w$, then $\len(w) = \Omega(d_G(g,1))$ by the triangle inequality.
\Thm{th:semisimple} claims (among other things) that we can choose such a $w$
with $\len(w) = O(R)$ when $d_G(g,1) < R$, so that the length is optimal up
to a constant factor.  We call the resulting algorithm \emph{long-distance
golf}, referring to the situation in golf where the target hole is much
further away than the longest-distance golf stroke.

Our version of long-distance golf in $G$ rests on two ideas.  The first idea
is to simplify the geometry from $G$ to $P \times E$ using \Thm{th:pkbilip}.
The second idea is to shorten a geodesic $\gamma$ in $P \times E$ with a
step $s$ of distance $\Theta(1)$ and word length $O(1)$ that approximately
follows $\gamma$.

We use the maps
\[ G \stackrel{\pi_Q}{\longto} Q
    \stackrel{\mu_Q^{-1}}{\longto} P \times E \]
to pass from distances in $G$ to distances in $P \times E$.  We can
abbreviate
\begin{align*}
d_{P \times E}(g,h) \defeq d_{P \times E}(\mu_Q^{-1}(\pi_Q(g)),
    \mu_Q^{-1}(\pi_Q(h))) \\
\end{align*}
and thus interpret the metric on $P \times E$ as a pseudo-metric on $G$.
We can calculate the metric $d_{P \times E}$ using \Prop{p:pdist} and
\equ{e:pythag}.

\begin{lemma} Given $w,g \in G$ with $d_{P \times E}(g,w) \ge 2$,
there is an $A$-word $s$ such that $\len(s) = O(1)$ and
\[ d_{P \times E}(g,ws) \le d_{P \times E}(g,w) - 1. \]
\label{l:longstep} \end{lemma} \eatline

\begin{remark} A version of \Lem{l:longstep} with the constants
1 and 2 replaced by $R$ and $2R$ for some $R > 0$ follows solely from
the fact that $\mu_Q^{-1} \circ \pi_Q$ is a quasiisometry.
\end{remark}

\begin{proof} We can choose $h \in G$ which is on the connecting $(P \times
E)$-geodesic at a distance of 2 from $w$ in $P \times E$;
\[ d_{P \times E}(h,w) = 2 \qquad
    d_{P \times E}(g,h) = d_{P \times E}(g,w) - 2. \]
Given that $d_Q(h,w)$ is the minimum distance between the coset $hC$
and $w \in G$ and given \Cor{c:pebilip}, we can choose $h$ such that
\[ d_G(h,w) = d_Q(h,w) = O(1). \]
Since the gate set $A$ densely generates $G$, we can choose $s$
with $\len(s) = O(1)$ and
\[ d_G(h,ws) = d_G(w^{-1}h,s) < \eps \]
for any $\eps > 0$.  In particular, since $\pi_Q$ is 1-Lipschitz
while $\mu_Q^{-1}$ is Lipschitz by \Cor{c:pebilip},
we can choose $\eps$ small enough to guarantee that
\[ d_{P \times E}(h,ws) < 1. \]
Thus
\[ d_{P \times E}(g,ws) < d_{P \times E}(g,w) - 1 \]
by the triangle inequality, as desired.
\end{proof}

\begin{figure}
\begin{ctp}[thick]
\draw (0,0) circle (4);
\draw (-1.6,1.1) node[rotate=10] {\includegraphics[height=.5cm]{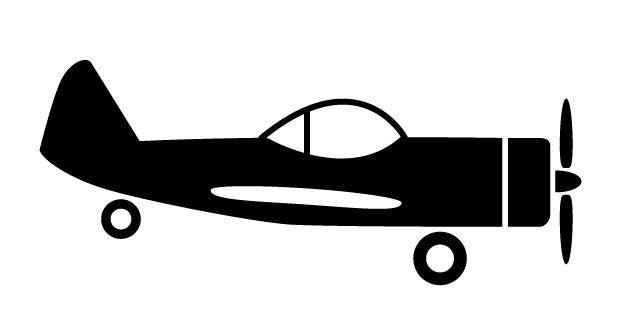}};
\draw (0.8,2.4) node {\includegraphics[height=1cm]{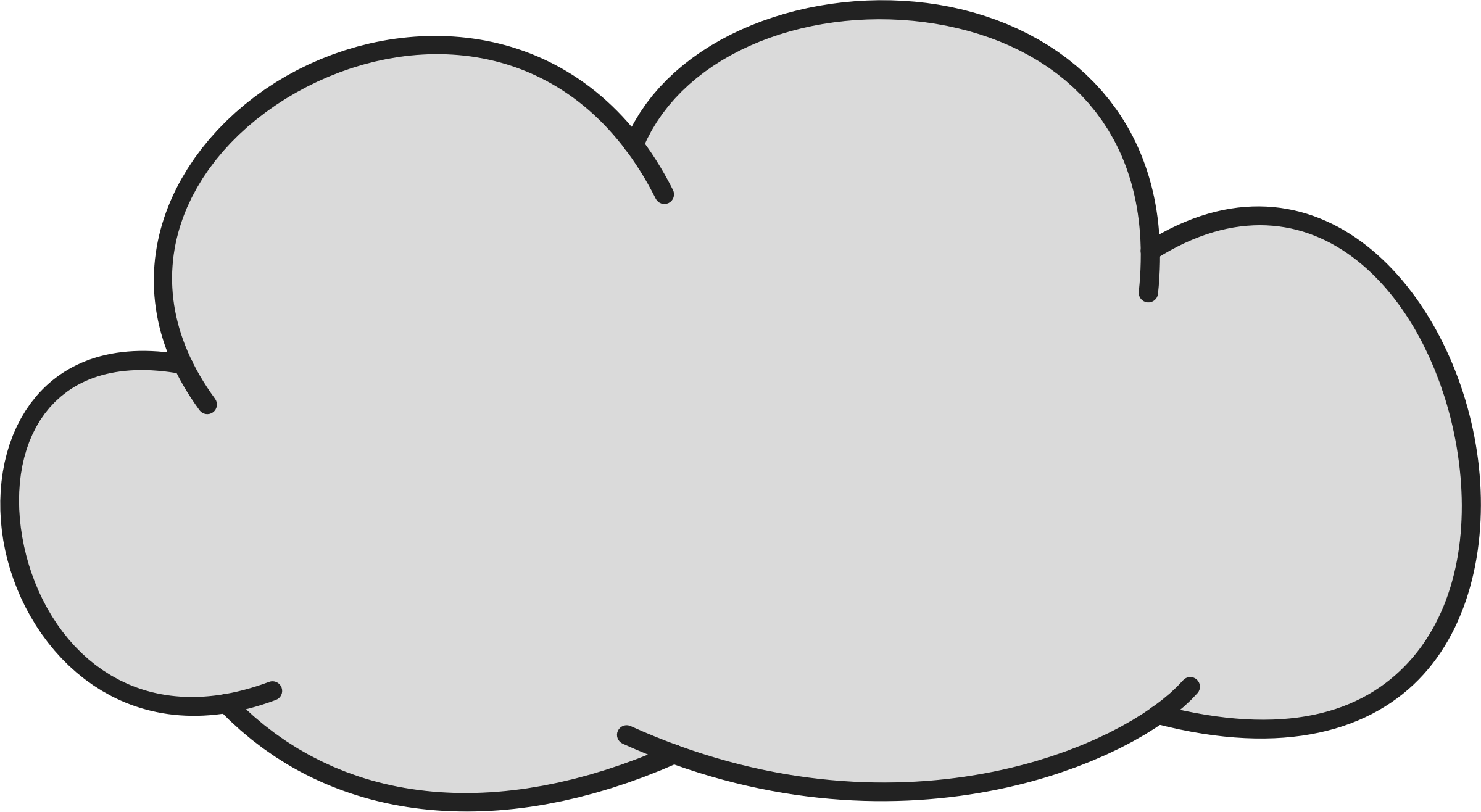}};
\draw (-2.34,-1.8) node {\includegraphics[height=.55cm]{golfer.png}};
\draw[->,gray] (-2,-2) arc (-9.5:9.5:2.432);
\draw[->,gray] (-2,-2) arc (110.9:97.2:3.46);
\draw[->,gray] (-2,-2) arc (81.9:61.3:1.768);
\draw[darkred] (-2,-2) arc (-31.0:-22.2:5.832);
\draw[darkred] (-1.6,-1.2) arc (117.9:115.2:19.236);
\draw[darkred] (-0.8,-0.8) arc (94.6:85.4:10.032);
\draw[darkred] (0.8,-0.8) arc (152.4:140.2:6.816);
\draw[->,darkred] (1.6,0.4) arc (-91.5:-88.5:15.004);
\draw[darkblue,dashed] (-2.0,-2.0) arc (128.3:101.2:12.204);
\fill (-2.0,-2.0) circle (0.06) node[below,inner sep=3pt] {\small 1};
\fill (-1.6,-1.2) circle (0.06) node[above,inner sep=3pt] {$w_1$};
\fill (-0.8,-0.8) circle (0.06) node[above,inner sep=3pt] {$w_2$};
\fill (0.8,-0.8) circle (0.06) node[below,inner sep=3.5pt] {$w_3$};
\fill (1.6,0.4) circle (0.06) node[above,inner sep=3pt] {$w_4$};
\fill (3.2,0.4) circle (0.06) node[right,inner sep=3pt] {$g$};
\fill[darkred] (3.2,0.95) -- (3,0.87) -- (3.2,0.79);
\draw (3.2,0.4) -- (3.2,0.975);
\draw (1.4,-2.4) node {\large $\H^2$};
\tikzset{scale=4};
\end{ctp}
\caption{Greedy long-distance golf in $Q \cong P \times E$, here $\H^2$.}
\label{f:lgolf}
\end{figure}

\Lem{l:longstep} yields a greedy algorithm, \Alg{a:l} below, to find an
$A$-word $w$ such that $d_G(w,g) = O(1)$; see \Fig{f:lgolf}.  The resulting
word $w$ has length $\len(w) = O(d_G(g,1))$ as $d_G(g,1) \to \infty$,
which is optimal up to a constant factor.  We finish the algorithm with
an exhaustive search to improve the upper bound on $d_G(g,w)$ from a
happenstance constant to a desired target value.

\begin{algorithm}{L} The input is a finite set
\[ A = A^{-1} \subseteq G \]
that densely generates a semisimple Lie group $G$, a target element $g \in
G$, a rational target precision $t \ge 0$, and an integer constant $b > 0$
(depending only on $A$ and $G$) that must be chosen favorably.  The output
is an $A$-word $w$ such that $d_G(g,w) < 2^{-t}$ and $\len(w) = O(d_G(g,1))$.
\begin{enumerate}
\item Set $w := 1$.
\item While $d_{P \times E}(g,w) \ge 2$:
\begin{enumerate}
\item Search for an $A$-word $s$ with $\len(s) \le b$ such that
\[ d_{P \times E}(g,ws) < d_{P \times E}(g,w) - 1. \]
\item Set $w := ws$.
\end{enumerate}
\item Search for an $A$-word $s$ with $\len(s) \le b$ and with the aid of
\equ{e:gnest}, such that $g^{-1}ws$ satisfies
\[ d_G(g,ws) = d_G(g^{-1}ws,1) \le 2\norm{g^{-1}ws-1}_2 \le 2^{-t}. \]
Return $ws$.
\end{enumerate}
\label{a:l} \end{algorithm}

The main part of \Alg{a:l}, namely the loop in step 2, is justified by
\Lem{l:longstep}. By \Cor{c:pebilip} and \equ{e:gqdist}, the iteration
in step 2 stops when
\[ d_{P \times E}(g,w) < 2 \quad \implies \quad d_G(g,w) = O(1).\]
After that, step 3 is an exhaustive search that uses \equ{e:gnest}, as the
algorithm says.  This becomes expensive as $t \to \infty$, but \Alg{a:l}
will only be used for a bounded value of $t$.  We will still use \Alg{a:l}
when $G$ is compact for the sake of step 3, even though in this case step
2 is trivial.

We can analyze \Alg{a:l} using similar reasoning as the analysis of the
algorithms in \Sec{s:su2}, except that there is no particular advantage
to using compressed words.  The algorithm takes $O(R)$ steps.  We must
be more careful with numerical precision because the metric on $G$ is not
bi-invariant when $G$ is not compact.  Because $d_G(s,1) = O(1)$ for each
step $s$, the first relation in \eqref{e:bmul} tells us that each step of
\Alg{a:l} sacrifices $O(1)$ bits of precision.  Because $\len(s) = O(1)$
as well, it suffices to give each gate $O(R)$ bits of precision beyond the
$t+O(\log(t))$ bits of precision that \Alg{a:z} will need, for a total
precision of $t+O(\log(t)+R)$ bits.  Since\Alg{a:z} will use $w^{-1}g$
as its target, the initial target $g$ needs $t+O(R)$ bits of precision.
The total time complexity of \Alg{a:l} is then $\tO(R(R+t))$.  We can
simplify this to $\tO((R+t)^2)$, because the time complexity formula for
\Alg{a:z} subsumes an extra $\tO(t^2)$ term.

If the Lie group $G$ and the gate set $A$ are both algebraic, then
we can improve the time complexity with a divide-and-conquer algorithm
that is typical for fast arithmetic.  Note that $E$ is trivial
and can be omitted in this case.

\begin{algorithm}{LD} The input is a finite set
\[ A = A^{-1} \subseteq G \]
of algebraic gates that densely generate a semisimple algebraic group $G$,
a target element $g \in G$, an optional rational target precision $t \ge 0$,
and an integer constant $b > 0$ (depending only on $A$ and $G$) that must
be chosen favorably.  The output is an $A$-word $w$ such that $d_G(g,w)
< 2^{-t}$ and $\len(w) = O(d_G(g,1))$.
\begin{itemize}
\item Calculate $p = \sqrt{gg^*}$.
\item If $d_P(g,1) > 4$, recursively call this algorithm with the target
$g' := \sqrt{p}$ and with no target precision to obtain a word $w_1$.
Recursively call this algorithm a second time with the target $g' :=
w_1^{-1}g$ to obtain a word $w_2$, and let $w := w_1w_2$.
\item If $d_P(g,1) \le 4$, search for an $A$-word $w$ with $\len(w) \le b$
such that $d_P(g,w) < 1$.
\item If the target precision $t$ is not present, return $w$.
\item If $t$ is present, search for an $A$-word $s$ with $\len(s) \le b$
and with the aid of \equ{e:gnest}, such that $g^{-1}ws$ satisfies
\[ d_G(g,ws) = d_G(g^{-1}ws,1) \le 2\norm{g^{-1}ws-1}_2 \le 2^{-t}. \]
Return $ws$.
\end{itemize}
\label{a:ld} \end{algorithm}

\begin{remark} It should be possible to generalize \Alg{a:ld} to any
semisimple $G$, assuming then that $A$ is a lift of an algebraic gate
set in $G_{\min}$.  However, the details are more complicated when $E$
is nontrivial, and it is arguably a less natural case.
\end{remark}

The recursion in \Alg{a:ld} has logarithmic depth.
Although $\len(w) = O(R)$, each word $w'$
at the $k$th-to-last stage of the recursion has length
\[ \len(w') \le 2^{-k}NR \]
for some constant $N$.  Thus we can apply \equ{e:bitlen} to conclude
that the total bit complexity of all of the values $(w')_G$ is $\tO(R)$.
\Alg{a:ld} also only needs $O(R)$ bits of precision of the target $g$,
and not the extra $t+O(\log(t))$ bits used later in \Alg{a:z}.  Similarly,
each target $g'$ at the $k$th-to-last stage of the recursion requires at
most $2^{-k}NR$ bits of precision for some constant $N$.  Thus the total
bit complexity of all of the numbers in the algorithm is $\tO(R)$, and
thanks to fast arithmetic the total time complexity is also $\tO(R)$.

\subsection{Locally surjective zigzags}
\label{ss:surject}

The algorithm in \Sec{ss:zgolf} uses the fact that
\[ d_{\SU(2)}(g,1) = d_{\SU(2)}(h,1) \]
for two elements $g,h \in \SU(2)$ if and only if they are conjugate.
Among semisimple Lie groups $G$, the ``if'' direction is only true when
$G$ is compact and thus has a bi-invariant metric, while the ``only if''
direction is only true when $G$ is $\SU(2)$ or $\SO(3)$.  Nonetheless,
using \Lem{l:zigzag} below, a generalized zigzag strategy works in any
semisimple Lie group $G$, with these changes:
\begin{enumerate}
\item We use $m$ zigs and $m$ zags to accurately reach a target element
$g \in G$, where $m$ is some integer that depends only on $G$.
\item A step $s$ might not be conjugate to its inverse $s^{-1}$.  We use
conjugates of both $s$ and $s^{-1}$ to make the zigzags.
\item If $G$ is semisimple but not simple, then a roughly exponential
step $s$ must be appropriately scaled in each simple factor of $G$.
\end{enumerate}

Algorithms~\ref{a:l} and \ref{a:ld} in \Sec{ss:long} produce a word $w$
such that
\[ d_G(w^{-1}g,1) = d_G(g,w) < 2^{-t} \]
for any target precision $t$.  The element $w^{-1}g \in G$ becomes a new
target for the short-distance \Alg{a:z}.  By setting $t$ to a high enough
value that depends only on $G$, we can assume that $g'$ lies in a ball
$B_G(2^{-t},1)$ that also embeds in $G_{\min}$.  Thus we can henceforth
assume that $G = G_{\min}$ and that $G$ factors according to \equ{e:gprod}.
Given
\[ g = (g_1,g_2,\ldots,g_q) \in G \qquad h = (h_1,h_2,\ldots,h_q) \in G \]
with $g_j,h_j \in G_j$, we define a family of pseudo-metrics (using the
left-invariant metric $d_{G_j}$ defined in \Sec{ss:metrics} for each factor
$G_j$) to address the third concern in the above list:
\[ d_{G,j}(g,h) \defeq d_{G_j}(g_j,h_j). \]

\newcommand{\hh}{\hat{h}}

We empower our more complicated zigzags with the following lemma.
To set up the lemma, let
\[ G = G_1 \times G_2 \times \cdots \times G_q \]
be a product of simple Lie groups as in \equ{e:gprod}, let $R,\eps > 0$,
and let $s \ne 1 \in G$ be a nontrivial element such that
\[ d_{G,j}(s,1) = \Theta(\eps) \]
for all $j$.  Let $m \ge 1$, and let
\begin{eq}{e:zzs} h(\vu,\vv) = s^{u_1} (s^{-1})^{v_1} s^{u_2}
    (s^{-1})^{v_2} \cdots s^{u_m} (s^{-1})^{v_m} \end{eq}
be a general product of $m$ zigzag pairs using conjugations of $s$ and
$s^{-1}$ by $u_j,v_j \in B_G(R,1)$.  Now fix $(\vu,\vv)$.  Using the change
of coordinates
\[ u'_j = u_j\exp(x_j) \qquad v'_j = v_j\exp(y_j), \]
define $\hh:B_L(\delta,0)^{2g} \to L$ by
\[ \hh(\vx,\vy) = \log(h(\vu',\vv')) \]
for $\delta > 0$ small enough that the logarithm is single-valued.  Then:

\begin{lemma} Let $G$, $R$, $\eps$, $s$, $m$, and $\hh$ be as above.
Then there is a value of $m$ (depending only on $G$)
such that if $\eps$ is small enough, then:
\begin{enumerate}
\item The function $h(\vu,\vv)$ has a bounded first derivative:
\[ \norm{dh(\vu,\vv)}_\infty = O(\eps) \]
for all $(\vu,\vv) \in B_G(2R,1)$.
\item There is a radius $r = \Omega(\eps)$ (depending only on $G$) and a
favorable point $(\vu,\vv) \in B_G(R/2,1)^{2m}$ (depending on $G$ and $s$)
and such that: For every $h_0 \in B_G(r,1)$, the equation
\[ \hh(\vx,\vy) = \log(h_0) \]
can be solved for $(\vx,\vy)$ using Newton's method initialized at
$(\vec0,\vec0)$. Moreover, the resulting solution for $(\vu',\vv')$ lies
in $B_G(R,1)^{2m}$.
\end{enumerate}
\label{l:zigzag} \end{lemma}

Conceptually, \Lem{l:zigzag} says that more complicated zigzags in a
general semisimple $G$ have the same numerical precision and stability
properties as the acute triangles in \Alg{a:zb}.  Likewise in the proof,
the pairs of steps $s^u$ and $(s^{-1})^v$ are at an acute angle of size
$O(\delta)$, where $\delta$ is small but independent of $\eps$.

In the lemma statement, although the system of equations for $(\vx,\vy)$
is underdetermined, Newton's method can be generalized to this case using
the Moore-Penrose pseudoinverse of $d\hh$ \cite{BenIsrael:system}.

We use the following sublemma in the proof \Lem{l:zigzag}.  The sublemma
statement involves the condition number $\kappa(f)$ of a linear map
$f:V \to W$ between two finite-dimensional real inner product spaces,
by definition the ratio of the largest and smallest singular values:
\[ \kappa(f) \defeq \frac{\sigma_{\max}(f)}{\sigma_{\min}(f)}. \]
Note that $\kappa(f) = \infty$ when $f$ does not have maximal rank.

\begin{lemma} Let $G$ be a simple algebraic group with Lie algebra $L$,
let $R > 0$, and let $m = \dim(G)$.  Given $z \in L$, let
\[ Z:\R^m \to L \qquad Z(\ve_j) = z^{g_j} \]
be the given linear map, where $(\ve_j)$ is the standard basis of
$\R^m$.  Then for all $z$ with unit norm $\norm{z}_L = 1$, there exists
\[ \vg = (g_1,g_2,\ldots,g_m) \in B_G(R,1)^m \]
such that $Z$ is invertible, and the condition number $\kappa(Z)$ is
bounded as a function of $z$.
\label{l:zzbasis} \end{lemma}

\begin{proof} Let
\[ S_L \defeq \{z \in L \st \norm{z}_L = 1 \} \]
be the unit sphere in $L$, and define the functions
\begin{align*}
f_z&:G^m \to \R & f_z(\vg) &\defeq \det(Z) \\
f&:S_L \to \R & f(z) &\defeq \max_{\vg \in B_G(R,1)^m} f_z(\vg).
\end{align*}
Since $G$ is simple, $L$ is irreducible as the adjoint representation of $G$.
Thus  $L$ is the span of all $z^g$ with $g \in G$ for any fixed $z \in S_L$.
Some subset $z^{g_1},z^{g_2},\ldots,z^{g_m}$ is then a basis of $L$, which
means that $f_z$ is not identically zero on its domain $G^m$.  Since $f_z$
is also real analytic, it is not identically zero on the domain $B_G(R,1)^m$;
and note also that $f_z$ is anti-symmetric as a function of its arguments.
Therefore the maximum value $f(z)$ of $f_z$ is positive for every $z \in
S_L$.  Since $S_L$ is compact and $f$ is continuous, we learn that $f(z)
= \det(Z)$ is bounded away from zero for the most favorable choice of $\vg
\in B_G(R,1)^m$ as a function of $z \in S_L$.

On the other side, each value $Z(\ve_j) = z^{g_j}$ is uniformly bounded
in norm when $z \in S_L$ and $g_j \in B_G(R,1)$.  Thus, $\norm{Z}_\infty =
\sigma_{\max}(Z)$ is also bounded.  Since $Z$ is a fixed-sized matrix, the
two inequalities together tell us that $\sigma_{\min}(Z)$ is bounded below
for the most favorable choice of $\vg \in B_G(R,1)^m$ as a function of $z$.
Thus the condition number $\kappa(Z)$ is bounded on $S_L$.
\end{proof}

\begin{proof}[Proof of \Lem{l:zigzag}] We first reduce to the case when $G$
is simple.  If $G$ is more generally semisimple, and if the lemma is true
for simple Lie groups, if we find solutions in each simple factor $G_j$,
then we can give them a common value of $m$ by padding \eqref{e:zzs} with
trivial factors of $ss^{-1}$ as needed.  The hypotheses of the lemma are
chosen so that we can combine these solutions to make a solution for the
product $G$.

Henceforth, we assume that $G$ is simple with Lie algebra $L$, and we
let $m = \dim(G)$ (which is not meant as an optimal value).  We let
\begin{eq}{e:sxeps} s = \exp(z) \qquad \norm{z}_L = \eps
    \qquad d_G(s,1) = \Theta(\eps). \end{eq}
Also let
\[ (\vu,\vv) \in B_G(R,1)^{2m} \qquad (\vx,\vy) \in L^{2m}
    \qquad \norm{x_j}_L,\norm{x_j}_L \le \delta. \]

We claim that the $k$th derivative of $\hh$ satisfies the norm bound
\begin{eq}{e:dkhh} d^k\hh(\vx,\vy) = O(\eps) \end{eq}
for each $k \ge 1$, if $\delta > 0$ (depending only on $G$ and $m$) is
small enough.  Here we interpret the $k$th derivative as a linear map
\[ d^k\hh(\vx,\vy):(L^{2m})^{\otimes k} \to L, \]
and then use the usual operator norm $\norm{\cdot}_\infty$ of a linear
map between two inner product spaces.

\Equ{e:dkhh} follows from the BCH formula \ref{e:bch} iterated $O(m)$ times.
We illustrate the reasoning with the simplest case $m=1$ so that we can
omit indices.  Given
\[ \qquad u' = u\exp(x) \qquad v' = v\exp(y) \qquad
    \norm{x}_L , \norm{y}_L \le \delta, \]
we obtain
\begin{align*} \hh(x,y) &= \log(s^{u'}(s^{-1})^{v'}) \\
    &= \log\bigl(\exp(z)^{u\exp(x)}\exp(-z)^{v\exp(y)}\bigr) \\
    &= \log\bigl(\exp(z^u)^{\exp(x^u)}\exp(-z^v)^{\exp(y^v)}\bigr) \\
    &= z^u+[x^u,z^u]-z^v-[y^v,z^v]-\frac{[z^u,z^v]}2
        + O(\eps\delta(\eps+\delta))
\end{align*}
as $\eps,\delta \to 0$. In context, the error term
$O(\eps\delta(\eps+\delta))$ has a stronger meaning than usual, because
the BCH formula is a power series with a radius of convergence, so we are
allowed to differentiate the error term.  If we differentiate with respect to
$(x,y)$ by taking $\delta \to 0$ and assume that $\eps$ is small, we obtain
\[ d\hh(0,0) = [dx,z]^u-[dy,z]^v + O(\eps^2) = \Theta(\eps). \]
If we differentiate $k$ times for any $k$, we obtain \equ{e:dkhh}, because
all terms in the power series for $\hh(x,y)$ that have $x$ or $y$ also
have at least one factor of $z$.  The same reasoning applies when $m >
1$, where will also use the more explicit first derivative expression
\begin{eq}{e:dh} d\hh(\vec0,\vec0) = \sum_{j=1}^m
    \left([dx_j,z]^{u_j} - [dy_j,z]^{v_j}\right) + O(\eps^2).
\end{eq}

Claim 1 is essentially a special case of \equ{e:dkhh}.
Since $G$ has a left-invariant metric, there is a left-invariant
isomorphism $T_g(G) \cong L$ for all $g \in G$.  Thus we can
interpret $dh$ as a linear map
\[ dh(\vu,\vv):L^{2m} \to L. \]
The definition of $\hh$ using $(\vu,\vv)$ then tells us that
\[ dh(\vu,\vv) = d\hh(\vec0,\vec0) \]
as linear maps, so they are both $O(\eps)$.

Claim 2 is more involved.  In outline, we use derivative estimates to
establish the implicit function theorem with a lower bound on the solution
radius.  The same bounds establish the convergence of Newton's method,
which is closely related.

As in \Lem{l:zzbasis}, let $S_L$ be the unit sphere in $L$ with respect
to the norm $\norm{\cdot}_L$, so that $z \in \eps S_L$.  Because $L$
is simple, there is a $y \in L$ such that $\norm{[y,z]}_L > 0$.  Moreover,
\[ f:S_L \to \R \qquad f(x) \defeq \max_{y_0 \in S_L} \norm{[y_0,z]}_L \]
is a continuous function of $z$, which implies that $f(z)$ is bounded away
from 0 on $\eps S_L$ as well as a bounded function.  We choose $y_0 \in S_L$
to maximize $\norm{[y_0,z]}_L$ and we let
\[ z' = \frac{[y_0,z]}{\norm{[y_0,z]}_L}. \]
We apply \Lem{l:zzbasis} with $z$ replaced by $z'$, and thus choose
\[ g_1,g_2,\ldots,g_m \in B_G(R,1) \]
to minimize the condition number $\kappa(Z)$ of the corresponding linear
map $Z$.  Since $z'$ is scaled to eliminate $\eps$ and $Z$ is defined from
$z'$, we can say that $\kappa(Z) = \Theta(1)$ as $\eps \to 0$.

We want to relate $\kappa(Z)$ to $\kappa(d\hh)(\vec0,\vec0)$ for a favorable
choice of $(\vu,\vv)$, namely $(\vu,\vv) = (\vg,\vg)$.  For the moment,
we also consider a special case of the perturbation $(\vx,\vy)$, namely:
\begin{align*}
\vx &= \vec0 \\
\vy &= (a_1y_0,a_2y_0,\ldots,a_my_0) = \va y_0 \\
\va &= (a_1,a_2,\ldots,a_m) \in \R^m,
\end{align*}
We can let $Y_0 \subseteq L^{2m}$ be the span
of all vectors $(\vec0,\va y_0)$, so that
\[ \hh|_{Y_0} \defeq \hh(\vec0,\va y_0) \]
Then
\[ d\hh|_{Y_0}(\vec0) = \norm{[y_0,z]}_L \bigl(Z(d\va) + O(\eps)\bigr). \]
Since $\norm{[y_0,z]}_L = \Theta(\eps)$ and $\kappa(Z) = \Theta(1)$,
we learn that
\[ \kappa(d\hh|_{Y_0}) = \Theta(1) \qquad
    \sigma_{\max}(d\hh|_{Y_0}), \sigma_{\min}(d\hh|_{Y_0}) = \Theta(\eps). \]
To relate this to $d\hh$ itself, in general if $f:V \onto W$ is a
surjective linear transform which is still a surjection on a subspace $V'
\subseteq V$, then
\[ \sigma_{\min}(f) \ge \sigma_{\min}(f|_{V'}). \]
We can apply this in our situations with
\[f = d\hh(\vec0,\vec0) \qquad V = L^{2m} \qquad V' = Y_0, \]
to conclude (together with the upper bound in claim 1) that
\begin{align*}
\sigma_{\min}(d\hh(\vec0,\vec0)) &= \Omega(\eps) \\
\sigma_{\max}(d\hh(\vec0,\vec0)) &=
    \norm{d\hh(\vec0,\vec0)}_\infty = O(\eps) \\
\kappa(d\hh(\vec0,\vec0)) &= \Theta(1).
\end{align*}

Finally for claim 2, the three equations
\begin{align*}
\norm{d\hh(\vec0,\vec0)}_\infty &= \Theta(\eps) \\
\kappa(d\hh(\vec0,\vec0)) &= \Theta(1) \\
\norm{d^2\hh(\vx,\vy)}_\infty &= O(\eps)
\end{align*}
are exactly those that establish the implicit function theorem and the
convergence of Newton's method in a region $B_L(\delta,0)^{2m}$ with $\delta
= \Omega(1)$, and the image of this region under $\hh$ contains $B_G(r,1)$
with $r = \Omega(\eps)$.
\end{proof}

\subsection{Roughly exponential short steps}
\label{ss:reshort}

In this subsection, we generalize \Alg{a:sb} to the case of an arbitrary
semisimple Lie group $G$.  Our generalization of \Alg{a:sb} is simpler when
the higher commutator $\omega$ is UCC.  We focus first on the UCC special
case, which by \Lem{l:elkucc} includes the cases when $G$ is compact and
$\omega = \omega_n$ is an Elkasapy word.  In particular, it includes the
case $G = \SU(d)$ which is stated as \Thm{th:qudit}.

\begin{lemma} Let
\[ G = G_1 \times G_2 \times \cdots \times G_q \]
be a product of simple Lie groups as in \equ{e:gprod}, let $R > 0$, and
let $\omega(g,h)$ be a UCC word in $G$ with
\[ \ccan_G(\omega) = n \ge 2. \]
Let $\eps > 0$, and then choose $s \in G$ such that $d_j(s,1) = \Theta(\eps)$
for all $j$.  Then the range over all $u \in B_G(R,1)$ of the list of
distances
\[ (d_{G,1}(\omega(s,s^u),1),d_{G,2}(\omega(s,s^u),1),
    \ldots,d_{G,q}(\omega(s,s^u),1)) \]
includes a cube $[0,r]^q$ with $r = \Omega(\eps^n)$, in the limit
as $\eps \to 0$.  Moreover, $\omega(s,s^u)$ is $O(\eps^n)$-Lipschitz as
a function of $u$.
\label{l:rangeucc} \end{lemma}

Note that could have stated \Lem{l:rangeucc} for words $\omega$ with more
than two arguments; we will only need the two-argument case.

\begin{proof} We again reduce the existence statement to the case when $G$
is simple.  If $G$ is more generally semisimple, and if the lemma is true
for simple Lie groups, then we can find another radius $R'$
such that
\[ B_1(R',1) \times B_2(R',1) \times \cdots
    \times B_q(R',1) \subseteq B_G(R,1), \]
where
\[ B_j(R',1) = B_{G_j}(R',1) \subseteq G_j \]
for each $j$.  (Note that all of these radii are independent of $\eps$.)
Then the lemma for each $G_j$ produces an interval $[0,r_j]$, and we can
let $r = \min_j r_j$.  Note that this reduction also uses the fact that
if $\omega$ is UCC in $G$, then it also is in $G_j$ and the conjugate
cancellation degrees are all the same.

Now assume that $G$ is simple, and choose $s$, $x$, and $\eps > 0$ such that
\[ s = \exp(\eps x) \qquad \norm{x}_L = 1.\]
We can simplify the uniform conjugate cancellation formula \eqref{e:ucc}
to the statement that for every $x \in L \setminus \{0\}$, there is a $u
\in G$ such that
\[ \omega^{(n)}(x,x^u) \ne 0. \]
Fixing $x$ for the moment, $\omega^{(n)}(x,x^u)$ is real analytic in $u$,
so that it is nonzero for a generic $u \in B_G(R,1)$ if it is ever nonzero
for any $u \in G$.  By continuity and compactness,
\[ m \defeq \min_{\norm{x}_L = 1}\biggl[\max_{u \in B_G(R,1)}
    \norm{\omega^{(n)}(x,x^u)}_L\biggr] > 0.  \]
Since $\ccan \omega \ge 2$, $\omega(s,s) = 1$ identically, which implies that
$\omega^{(n)}(x,x) = 0$.  Thus $\norm{\omega^{(n)}(x,x^u)}_L$ achieves both
$0$ and $h$ as $u \in B_G(R,1)$ varies, for each $x$ with $\norm{x}_L = 1$.
Since $B_G(R,1)$ is connected, $\norm{\omega^{(n)}(x,x^u)}_L$ achieves the
entire interval $[0,m]$.  Since
\begin{eq}{e:dnomega} \omega(\exp(\eps x),\exp(\eps x^u))
    = \exp(\eps^n \omega^{(n)}(x,x^u)+O(\eps^{n+1})), \end{eq}
we obtain
\[ d_G(\omega(s,s^u),1) =
    \eps^n \norm{\omega^{(n)}(x,x^u)}_L + O(\eps^{n+1}). \]
Considering the error term, we can thus set $r = \eps^n m/2$.

The proof of the Lipschitz statement is similar to the proof of the
corresponding statement in \Lem{l:zigzag}.  First, $x \mapsto x^u$ is
$O(1)$-Lipschitz.  Second, since \eqref{e:dnomega} comes from a power series
within its radius of convergence, we are allowed to differentiate within the
error term to conclude that the entire expression is $O(\eps^n)$-Lipschitz.
\end{proof}

We turn to the algorithm to produce roughly exponential steps for a general
semisimple $G$, assuming that the word $\omega$ is UCC.  As explained
at the end of \Sec{ss:long}, we can replace $G$ with $G_{\min}$ and thus
assume that $G$ is a product of simple Lie groups.

\begin{algorithm}{SU} The input is a finite set
\[ A = A^{-1} \subseteq G \]
that generates a product
\[ G = G_1 \times G_2 \times \cdots \times G_q \]
of simple groups, a UCC word $\omega(g,h)$ with
\[ n = \ccan_G(\omega) \ge 2, \]
and a rational target precision $t > 0$.  The algorithm also depends
on two integer constants $a,b > 0$ (depending only on $A$) that must be
chosen favorably.

The output is an $A$-word $s_t$ such that
\begin{eq}{e:ustep} 2^{1-t}  > d_{G,j}(s_t,1) > 2^{-t}. \end{eq}
\begin{enumerate}
\item If $t \le 2a$, search for an $A$-word $s_t$ with $\len(s_t) \le b$ that
satisfies the \equ{e:ustep}, and return this word.

\item If $t > 2a$, recursively (using this algorithm) calculate an
$A$-word $s_{t'}$ with
\[ t' := \frac{t-a}n. \]
Then search for an $A$-word $u$ such that $\len(u) \le b$, and such that
\[ s_t := \omega(s_{t'},s_{t'}^u) \]
satisfies \equ{e:ustep}.
\end{enumerate}
\label{a:su} \end{algorithm}

Our description of \Alg{a:su} does not have an explicit formula for the
distance $d_{G,j}(s_t,1)$ analogous to \Prop{p:cross} for \Alg{a:sb2}
or \Lem{l:trig} in \Alg{a:sb}.  However, the algorithm still works using
the more abstract \Lem{l:rangeucc}.  Given the Lipschitz property in
\Lem{l:rangeucc}, it suffices to choose $u \in B_G(R,1)$ with bounded
precision and thus $\len(u) = O(1)$.  \Alg{a:su} also produces a compressed
form $\Delta_t$ of $s_t$.  Using the parameters \eqref{e:elal}, the
recurrences \eqref{e:lenrec} and \eqref{e:clenrec} still hold.  Thus
\[ \len(s_t) = O(n^\alpha) \qquad \len(\Delta_t) = O(\log(t)) = \tO(1) \]
as before.

The time complexity and precision demands of \Alg{a:su} are also
asymptotically the same as those of \Alg{a:sb}.  If the algorithm is
implemented with generalized interval arithmetic, then all of the group
elements in $\SU(2)$ that it uses need $n+\Theta(n)$ bits of precision,
and the algorithm runs in time $\tO(n)$.  If instead the algorithm
is implemented with exact arithmetic over a number field $K$, then
$\norm{s_t}_{\bit} = O(t^\alpha)$ by \equ{e:bitlen}, and the algorithm
runs in time $\tO(t^\alpha)$.

To conclude this subsection, we consider the most general algorithm to
generate roughly exponential steps, in any semisimple $G$ and using any
word $\omega(g,h)$ with $\ccan \omega \ge 2$.  The algorithm addresses
the remaining two complications (listed as 3 and 4) at the beginning
of this section.  When $G$ has more than one simple factor $G_j$ and
$\ccan_{G_j}(\omega)$ is different in different factors, the remedy is
roughly exponential steps with multidimensional scaling.  If $\omega$
isn't UCC in some factor $G_j$, then we can desingularize it with a zigzag
product as in \Lem{l:zigzag}.  The remedies are implemented carefully so
that they do not affect the length exponent $\alpha$ in \eqref{e:elal}.

\begin{lemma} Let $G$ be a simple Lie group, let $R > 0$, and let
$\omega(g,h)$ with
\[ n = \ccan_G(\omega) \ge 2. \]
Then the range of $d_G(\omega(s,s^u),1)$ over all $u \in B_G(R,1)$ and $s
\in G$ such that $d_G(s,1) = O(\eps)$ includes an interval $[0,r]$ with $r =
\Omega(\eps^n)$, in the limit as $\eps \to 0$.  Moreover, $\omega(s,s^u)$
is $O(\eps^n)$-Lipschitz as a function of $u$ and $O(\eps^{n-1})$-Lipschitz
as a function of $s$.
\label{l:range} \end{lemma}

The proof of \Lem{l:range} is mostly similar to the proof of
\Lem{l:rangeucc}, except simpler.  The main difference is in the Lipschitz
estimate at the end, that $\omega(s,s^u)$ is $O(\eps^{n-1})$-Lipschitz
as a function of $s$.  If we let $s = \exp(\eps x)$ as in the proof
of \Lem{l:rangeucc}, then \equ{e:dnomega} says that $\omega(s,s^u)$
is $O(\eps^n)$-Lipschitz as a function of $x$.  However, $x$ itself is
$O(\eps^{-1})$-Lipschitz as a function of $s$ in a neighborhood of $x=0$.
Hence the composition yields a Lipschitz bound of $O(\eps^{n-1})$.

In the algorithm, we will replace the simple group $G$ in \Lem{l:range}
with a product of simple groups, and we will combine the lemma with
\Lem{l:zigzag}.   In other words, we will let $h$ be a zigzag product
as in \equ{e:zzs}, and then use the range and Lipschitz regularity of
$\omega(h,h^w)$.  Since $h$ itself is $O(\eps)$-Lipschitz as a function
of the conjugators $u_1,\ldots,u_m$ and $v_1,\ldots,v_m$, the composition
is $O(\eps^n)$-Lipschitz as a function of all of the conjugators.

\begin{algorithm}{S} The input is a finite set
\[ A = A^{-1} \subseteq G \]
that generates a product
\[ G = G_1 \times G_2 \times \cdots \times G_q \]
of simple groups, a word $\omega(g,h)$ with
\[ n_j = \ccan_{G_j}(\omega) \ge 2, \]
and a list of rational target precisions
\[ \vt = (t_1,t_2,\ldots,t_q) \qquad t_j > 0. \]
The algorithm also depends on two integers constants $a,b > 0$ (depending
only on $G$ and $A$) and an integer $m$ (depending only on $G$) that must
be chosen favorably.

The output is an $A$-word $s_{\vt}$ such that
\begin{eq}{e:jstep} 2^{1-t_j} > d_{G,j}(s_{\vt}) > 2^{-t_j} \end{eq}
for all $j$.

\begin{enumerate}
\item If $t'_j \le 2a$ for any $j$, search for an $A$-word $s_{\vt}$
with $\len(s_{\vt}) \le b$ that satisfies \equ{e:jstep} for all $j$,
and return this word.  If $t'_j > 2a$, follow the remaining steps.
\item Recursively (using this algorithm) calculate an $A$-word
$s := s_{\vt'}$ with
\[ \vt' := (t'_1,t'_2,\ldots,t'_q) \qquad t'_j := \frac{t'_j-a}{n_j}. \]
\item Search among lists of $A$-words
\[ u_1,u_2,\ldots,u_m,v_1,v_2,\ldots,v_m,w \]
with
\[ \len(u_j),\len(v_j),\len(w) \le b. \]
For each candidate list, set
\[ h := s^{u_1} (s^{-1})^{u_2} s^{u_3} (s^{-1})^{u_4} \cdots
    s^{u_{2m-1}} (s^{-1})^{u_{2m}}, \]
and stop the search when
\[ s_{\vt} := \omega(h,h^w) \]
satisfies \equ{e:jstep} for all $j$.  Return $s_{\vt}$.
\end{enumerate}
\label{a:s} \end{algorithm}

\Alg{a:s} uses \Lem{l:zigzag} to desingularize the inputs to $\omega$,
so it suffices to use a value of $m$ guaranteed by that lemma.

It is not possible to choose $a$ and $b$ for \Alg{a:s} to succeed as written
for all choices of the target precision list $\vt$.  If $\min_j(t_j)$ is
small, the algorithm will only execute step 1, but this step can only succeed
if $\max_j(t_j)$ is bounded.  This can be remedied with a more complicated
algorithm, but we will avoid this failure mode by guaranteeing that
\[ \max_j(t_j) \to \infty \quad \implies \quad \min_j(t_j) \to \infty \]
for all inputs to \Alg{a:s}, including recursive calls.  The outer calls
to the algorithm will always be of the form
\[ s_t := s_{(t,t,\ldots,t)}, \]
which implies that $t_j = \Theta(t_k)^{\Theta(1)}$ for all $j$ and $k$
throughout the recursion.

Recall that
\[ n = \ccan_G(\omega) = \min_j \ccan_{G_j}(\omega). \]
A length recurrence argument analogous to those for Algorithms~\ref{a:sb}
and \ref{a:su} tells us that
\[ \len(s_t) = O(t^\alpha) \qquad \alpha = \log_n(2m \ell), \]
that
\[ \len(\Delta_t) = O(\log(t)) = \tO(1) \]
for the compressed word $\Delta_t$ for $s_t$, and that the time complexity
and required numerical precision are the same as that of \Alg{a:su},
except with a larger value of $\alpha$.  The modified exponent
\[ \alpha = \alpha_{G,m}(\omega) = \frac{\log(\len(\omega))
    + \log(2m)}{\log(\ccan_G(\omega))} \]
is not as good as the exponent
\[ \alpha = \alpha_G(\omega)
    = \frac{\log(\len(\omega))}{\log(\ccan_G(\omega))} \]
from \Alg{a:su} when $\omega$ is UCC.  However, by the following lemma,
we claim the infimum of all values of $\alpha$ obtained in this way is
the same.  In particular, this claim combined with \Thm{th:elkasapy}
tells us that
\[ \inf_{\omega} \alpha_{G,m}(\omega) \le \log_\phi(2), \]
which we need for the full generality of \Thm{th:semisimple}.

\begin{lemma} If $\omega(g,h)$ is a nontrivial group word
with trivial abelianization, then there is a sequence
of nontrivial group words $\omega^{(n)}(g,h)$ such that
\[ \lim_{n \to \infty} \alpha_{G,m}(\omega^{(n)})
  = \lim_{n \to \infty} \alpha_G(\omega^{(n)}) \le \alpha_G(\omega). \]
\end{lemma}

The idea of the proof is to amplify $\omega$ by composing it with itself
so that the extra desingularization term matters less and less.

\begin{proof} The cancellation properties of $\omega(g,h)$ do not change
if we conjugate it in the free group $F_2$.  Thus we can find a conjugate
\[ \tomega(g,h) \sim \omega(g,h) \qquad \len(\tomega) \le \len(\omega) \]
which is cyclically reduced and begins with $g$.  It follows that
\[ \tomega(g,h)\tomega(g,h^{-1}) \ne \tomega(g,h^{-1})\tomega(g,h), \]
Nielsen--Schreier theorem tells us that $\tomega(g,h)$ and
$\tomega(g,h^{-1})$ generate a free group; since they do not commute,
they generate a free group of rank 2.  Therefore
\[ \omega^{(2)}(g,h) \defeq \omega(\tomega(g,h),\tomega(g,h^{-1})) \]
is nontrivial and satisfies
\[ \len(\omega^{(2)}) \le \len(\omega)^2 \qquad
    \ccan_G(\omega^{(2)}) \ge \ccan_G(\omega)^2. \]
If we repeat this construction to make
\[ \omega^{(n)}(g,h) \defeq
    \omega(\tomega^{(n-1)}(g,h),\tomega^{(n-1)}(g,h^{-1})), \]
words), we obtain the lemma statement.
\end{proof}

\subsection{The full algorithm}
\label{ss:full}

In this subsection, we give the full algorithm to establish
\Thm{th:semisimple}, one that works in any semisimple Lie group $G$
and that uses any higher commutator $\omega$.  \Thm{th:semisimple} then
follows if we choose $\omega$ using \Thm{th:elkasapy}.

\begin{algorithm}{Z} The input is a finite set
\[ A = A^{-1} \subseteq G \]
that densely generates a semisimple real Lie group $G$, a group word
$\omega(g,h)$ with $\ccan_{G}(\omega) \ge 2$, a target element $g \in G$,
and a rational target precision $t > 0$.  The algorithm also depends on a
rational constant $0 < \beta < 1$ and four positive integer constants $a$,
$b$, $c$, $m$, and $t_0$ that must all be chosen favorably and independently
of $t$.

The output is a word $w_t$ that satisfies \eqref{e:wtest}.
\begin{enumerate}
\item If $d_G(g,1) \ge 2^{-t_0}$, use Algorithm \ref{a:l} or \ref{a:ld}
with target precision $t_0$ to produce a word $w_0$ such that $d_G(g,w_0)
< 2^{-t'}$.  If $t \le t_0$, then return $w_0$.  If $t > t_0$, then set
\[ g := w_0^{-1}g \qquad G := G_{\min}, \]
and proceed to the remaining steps.
\item Let $t' := (1-\beta)t$, recursively calculate $w_{t'}$ using
this algorithm with input $(g,t')$ and set $g := w_{t'}^{-1}g$.
\item Calculate $s := s_{t'-c}$ using \Alg{a:su} or \Alg{a:s} and the
word $\omega$.
\item Solve the equation
\[ g = s^{g_{u,1}} (s^{-1})^{g_{v,1}} s^{g_{u,2}} (s^{-1})^{g_{u,2}} \cdots
   s^{g_{u,m}} (s^{-1})^{g_{v,m}} \]
for $g_j,h_j \in B_G(1,1)$.
\item Set $t'' := \beta t + a$. For each $1 \le j \le m$,
calculate $u_j$ and $v_j$ recursively as
\[ u_j := w_{u,j,t''} \qquad v_j := w_{v,j,t''} \]
using $(g_u,t'')$ and $(g_v,t'')$ as inputs to this algorithm.  Return
\[ w_t := w_{t'} s^{u_1} s^{v_1} s^{u_2} s^{v_2} \cdots s^{u_m} s^{v_m}. \]
\end{enumerate}
\label{a:z} \end{algorithm}

The constants in \Alg{a:z} can be set in the following order and manner:
\begin{enumerate}
\item The product length $m$ is provided by the statement and
proof of \Lem{l:zigzag}.
\item The constant $c$, which lengthens the steps $s$ by a factor of $2^c$,
should be set to exceed the scales $O(r)$ and $O(\delta)$ in \Lem{l:zigzag}
and its proof.
\item The constant $a$, which provided each conjugator $u_j$ and $v_j$ with
$2^a$ extra bits of precision, should be set so that the total precision
exceeds the Lipschitz constant of $h$ in the statement of \Lem{l:zigzag}.
\item We can set
\[ t_0 = \max\biggl(\frac{2a}{1-\beta},2c\biggr), \]
so that $t-t'$ and $t-t''$ are positive and bounded away
from 0 in the recursion in \Alg{a:z}.
\item Finally, we can choose $b$ based on the gate set $A$ at
each step when \Alg{a:z} needs a bounded-length $A$-word.
\end{enumerate}

To complete the proof of \Thm{th:semisimple}, we need to know that if
Algorithm~\ref{a:su} or \ref{a:s} produces a step $s_t$ of length
$O(t^\alpha)$, then \Alg{a:z} produces a word $w_n$ of length
\[ \len(w_t) = O(t^\alpha+R). \]
While \Sec{ss:long} addresses the long-distance $O(R)$ term, we need a
lemma analogous to \Lem{l:zb} for the short-distance term $O(t^\alpha)$ term.

\begin{lemma} Let $\alpha > 1$ be the word length exponent in
Algorithm~\ref{a:su} or \ref{a:s}, let $m$ be the given constant in
\Alg{a:z}, and suppose that the target element $g \in G$ in \Alg{a:z}
satisfies $d_G(g,1) \le 1$.  Let:
\[ \beta = (4m)^{1/(1-\alpha)} \qquad D =
    \frac{3m}{(1-\beta)^{1-\alpha}-1} \qquad N \ge 1. \]
Given these parameters, choose $N$ such that Algorithm~\ref{a:su} or
\ref{a:s} produces a word $s_t$ with $\len(s_t) \le Nt^\alpha$ for all
$t > 0$, and such that \Alg{a:z} produces a word $w_t$ with $\len(w_t)
\le DNt^\alpha$ when $t_1 \ge t \ge 1$. If $t_1$ is sufficiently large,
then it follows that $\len(w_t) \le DNt^\alpha$ for all $t$.
\label{l:z} \end{lemma}

\begin{proof} We generalize the proof of \Lem{l:zb}.

For all sufficiently large $t$, \Alg{a:z} produces:
\begin{gather*}
w_t = w_0 w_{t'} s^{u_1} (s^{-1})^{v_1} s^{u_2} (s^{-1})^{v_2}
    \cdots s^{u_m} (s^{-1})^{v_m} \\
s = s_{t'-c} \qquad u_j = w_{u,j,t''} \qquad v_j = w_{v,j,t''} \\
t' = (1-\beta)t \qquad t'' = \beta t + a
\end{gather*}
We proceed by induction on $t > t_1 \ge t_0$, using the fact that $t-t'$
and $t-t''$ are both bounded away from zero.
\begin{align*}
\len(w_t) &= \len(w_0) + \len(w_{t'}) + 2m\len(s_{t'-c}) \\
    &\qquad + 2\sum_j \bigl(\len(w_{u,j,t''})+\len(w_{v,j,t''})\bigr) \\
    &\le DNt_0^\alpha + DN(1-\beta)^\alpha t^\alpha
        + 2m N(1-\beta)^\alpha t^\alpha \\
    &\qquad + 4mDN(\beta t+a)^\alpha.
\end{align*}
We want to show that $\len(w_n) \le DNt^\alpha$.  Dividing through by
$Nt^\alpha$, we want to show that
\[ D \ge D(1-\beta)^\alpha + 2m(1-\beta)^\alpha+4mD\beta^\alpha +
    O(t^{-1}). \]
We apply the substitution $\beta = 4m \beta^\alpha$ and move all terms
with $D$ to the left side.  Then we want to show that
\[ D(1-\beta - (1-\beta)^\alpha) \ge 2m(1-\beta)^\alpha + O(t^{-1}). \]
Substituting the value of $D$, we want to show that
\[ 3m(1-\beta)^\alpha \ge 2m(1-\beta)^\alpha + O(t^{-1}), \]
which holds when $t > t_1 \gg 1$.
\end{proof}

As in \Sec{ss:zgolf}, we also need a generalization of \Lem{l:z} to analyze
the time complexity of \Alg{l:z} and the compressed word length of the
result, and thus complete the proof of Theorems~\ref{th:semisimple}
and \ref{th:runtime}.

\begin{lemma} In the notation of \Lem{l:z}, let $0 < \beta' < \beta$.
If \Alg{a:z} uses the parameter $\beta'$ instead of $\beta$, then we
still obtain $\len(w_t) = O(t^\alpha)$.
\label{l:zp} \end{lemma}

The proof is the same as that of \Lem{l:zbp}.

\Sec{ss:long} and \Lem{l:z} together establish the main word length estimate
in \Thm{th:semisimple}, that $\len(w_t) = O(n^\alpha+R)$.  We turn to the
compressed length estimate in the same theorem.

If $\Delta_{w,t}$ is the compressed word form of $w_t$ as produced by
\Alg{a:z}, and similarly for the other words in that algorithm, we obtain
the compressed length recurrence
\begin{multline*}
\len(\Delta_{w,t}) = \len(w_0) + \len(\Delta_{w,t'}) + \len(\Delta_{s,t'-c}) \\
    + \sum_j \bigl(\len(\Delta_{u,j,t''})+\len(\Delta_{v,j,t''})\bigr)
\end{multline*}
in the bounded distance special case $d_G(g,1) \le 1$.  This recurrence
is dominated by the recurrence for $\len(w_t)$, and $\len(\Delta{s,t})
= \tO(1)$.  By the same argument as in \Sec{ss:zgolf}, we obtain
\[ \len(\Delta_{w,t}) = O(t^{1+\delta}) \]
for every $\delta > 0$, or
\[ \len(\Delta_{w,t}) = O(t^{1+\delta}+R) \]
with no restriction on $d_G(g,1)$.

Finally, we establish the precision and runtime estimates in
\Thm{th:runtime}.  The target $g$ needs $t+O(1)$ bits of precision,
while the gates in $A$ needs the same precision plus the precision lost
to group multiplication, which at each step is bounded by the inclusions
\eqref{e:bmul}.  If $d_G(g,1) \le 1$, then the recursion in \Alg{a:z}
(including \Alg{a:s} along with it) has logarithmic depth, and so does
the circuit $\Delta_{w,t}$.  Moreover, in this regime each step of
\Alg{a:z} has a bounded product of group elements, and each such group
element is a bounded distance from $1 \in G$.  Hence by \eqref{e:bmul},
each step loses only $O(1)$ bits of precision, so that the gates need
$t+O(\log(t))$ bits of precision.  Again by the same reasoning as in
\Sec{ss:zgolf}, the total runtime of the \Alg{a:z} when $d_G(g,1) \le 1$
is
\[ \len(\Delta_{w,t})\tO(t) = \tO(t^{2+\delta}) \]
for any $\delta >
0$. With no restriction on $d_G(g,1)$, we must add time complexity and
precision requirements of the first call \Alg{a:l}.  The target $g$ needs
$t+O(R)$ bits of precision, the gates in $A$ need $t+O(\log(t)+R)$ bits
of precision, and the total time complexity is $\tO(t^{2+\delta}+(R+t)^2)$.

If the gate set $A$ is algebraic, then as before we can do better with
exact arithmetic.  Combining the time complexity of \Alg{a:ld} with the
reasoning in \Sec{ss:zgolf} applied to \Alg{a:z}, the total time complexity
is $\tO(t^\alpha+R)$.

\section{Proof of \Thm{th:pkbilip}}
\label{s:bilip}

In this section, we establish \Thm{th:pkbilip} as a special case of
\Thm{th:bilip} given below.  The generalization also subsumes a lemma of
Chatterji, Pittet, and Saloff-Coste \cite[Lem.~8.7]{CPS:rd}.

To state the generalization, we need the concept of a \emph{pseudo-Riemannian
metric} on a connected, smooth manifold $M$.  By definition, it is
a continuous assignment of a nonsingular symmetric inner product
$\braket{\cdot,\cdot}_p$ to each tangent space $T_p(M)$.  This type of
tensor field on $M$ is called a ``metric'' even though it does not make
$M$ into a metric space in the usual sense (nor even a pseudo-metric
space) in the novel case when $\braket{\cdot,\cdot}_p$ is indefinite.
If $\braket{\cdot,\cdot}_p$ is either positive or negative definite,
then the pseudo-Riemannian metric amounts to a Riemannian metric.

Borrowing from the theory of relativity (which uses pseudo-Riemannian
metrics), a nonzero tangent vector $x \in T_p(M)$ is \emph{spacelike},
\emph{timelike}, \emph{non-spacelike}, \emph{non-timelike} or \emph{null}
when, respectively, $\braket{x,x}_p$ is positive, negative, nonpositive,
nonnegative, or zero.  Likewise, a submanifold $P \subseteq M$ can be
\emph{spacelike}, \emph{timelike}, \etc, when all nonzero tangent vectors
$x \in T_p(P)$ have the same property.

Let $G$ be a semisimple Lie group.  To state \Thm{th:bilip}, we use
two left-invariant pseudo-Riemannian metrics on $G$, both given by
symmetric inner products on the Lie algebra $L$.  As in \Sec{ss:nonmat}
and \Sec{ss:metrics}, let $\rho:G \to \SL(d,\C)$ be a matrix representation
which is faithful on $L$ and unitary on $K$.  Given $x \in L$, we write
it as $x = y+z$ with $y \in L_P$ and $z \in L_K$.  First, we have the
positive-definite inner product \eqref{e:hsnorm} given by
\begin{eq}{e:gplus} \begin{aligned}
\braket{x,x}_{G,+} &\defeq \tr(\rho(x)\rho(x)^*) \\
\braket{y,y}_P &\defeq \tr(\rho(y)\rho(y)^*) \\
\braket{z,z}_K &\defeq \tr(\rho(y)\rho(y)^*) \\
\braket{x,x}_{G,+} &= \braket{y,y}_P + \braket{z,z}_K
\end{aligned} \end{eq}
Second, we have the inner product
\begin{eq}{e:gminus} \braket{x,x}_{G,-} \defeq \Re(\tr(\rho(x)^2))
    = \braket{y,y}_P - \braket{z,z}_K. \end{eq}
The inner product \eqref{e:gminus}, which is indefinite when $G$ is
not compact, is $G$-adjoint invariant and thus yields a bi-invariant
pseudo-Riemannian metric on $G$. (If $\rho$ is adjoint, then
$\braket{\cdot,\cdot}_{G,-}$ is the Killing form mentioned in \Sec{ss:lie}.)

As discussed in \Sec{ss:metrics}, we give $P \times K$ the product Riemannian
metric $\braket{\cdot,\cdot}_{P \times K}$, where $\braket{\cdot,\cdot}_K$
is the $K$-bi-invariant Riemannian metric on $K$ induced by the bilinear
form $\braket{\cdot,\cdot}_K$ on $L_K$, and $\braket{\cdot,\cdot}_P$
is the $G$-invariant Riemannian metric on $P \cong G/K$ induced by the
bilinear form $\braket{\cdot,\cdot}_P$ on $L_P$.

\begin{theorem}  Let $\sigma:P \to G$ be a smooth embedding of $P$ in $G$
which is a right inverse to the projection $\pi_P:G \to P$. Suppose
also that $\sigma(P)$ is a nowhere-timelike smooth manifold with respect to
the pseudo-Riemannian metric $\braket{\cdot,\cdot}_{G,-}$.  Then the
map $\mu:P \times K \to G$ given by
\[ \mu(p,k) \defeq \mu(p)k \]
is a bi-Lipschitz diffeomorphism between $P \times K$ with the
product metric $\braket{\cdot,\cdot}_{P \times K}$ and $G$ with the metric
$\braket{\cdot,\cdot}_{G,+}$.
\label{th:bilip} \end{theorem}

\Thm{th:pkbilip} is the special case of \Thm{th:bilip} in which $\sigma$
is the identity on $P$.  Thus, we need the following lemma.

\begin{lemma} The radial submanifold $P \subseteq G$ is spacelike.
\label{l:radial} \end{lemma}

\begin{proof} We can (as elsewhere) replace $G$ by $G_{\min}$, and assume
that $G \subseteq \SL(d,\C)$ is algebraic and that $\rho$ is the identity.

Let $p \in P$ and let $y \in L \cong T_1(G)$ be such that
$py \in T_p(P)$.  Since $P$ lies completely in the Hermitian subspace of
$\M(d,\C)$, every tangent vector $py \in T_p(P)$ is also Hermitian as an
element of $\M(d,\C)$.  Therefore
\[ py = (py)^* = y^*p. \]
Since $p$ is positive and Hermitian, it has a unique positive, Hermitian
square root $p^{1/2}$.  So we can write
\[ p^{1/2}yp^{-1/2} = p^{-1/2}y^*p^{1/2} = (p^{1/2}yp^{-1/2})^*. \]
Although in general $y \notin L_P$, $y$ is conjugate to $p^{1/2}yp^{-1/2}$,
which is Hermitian and does lie in $L_P$.  In particular, $\braket{y,y}_{G,-}
\ge 0$; $y$ lies in the spacelike subspace $p^{-1/2}L_Pp^{1/2}$.
\end{proof}

\begin{remark} In their special case of \Thm{th:bilip}, Chatterji, Pittet,
and Saloff-Coste use the Iwasawa decomposition $G = NAK$, where $N$ is
nilpotent, $A$ is abelian, and $S = NA$ is solvable.  There is a canonical
diffeomorphism $\sigma_S:P \to S$ because $S \cong G/K \cong P$; \ie, both
$P$ and $S$ are sections of the quotient map $\pi_P:G \to G/K$.  By analogy
with \Lem{l:radial}, the solvable subgroup $S$ is nowhere-timelike because
it is null along the nilpotent subgroup $N$ and spacelike along the abelian
subgroup $A$.  Thus \Thm{th:bilip} shows that the Iwasawa multiplication map
$\mu(p,k) = \sigma_S(p)k$ is also bi-Lipschitz.  Chatterji, Pittet, and
Saloff-Coste instead observe that the differential (or derivative) $d\mu$
of the Iwasawa multiplication map is the same everywhere in the sense that
$\mu$ is invariant under left multiplication by $S$ and right multiplication
by $K$.  This immediately shows that their $\mu$ is bi-Lipschitz, but
(unlike in our argument) with no particular bound on the Lipschitz constants.
\end{remark}

To prove \Thm{th:bilip} we use the following lemma from linear algebra.

\begin{lemma} Consider $\R^{n+m}$ with the indefinite inner product
\[ \braket{\vx,\vy}_- = \sum_{j \le n} x_j y_j - \sum_{j > n} x_j y_j. \]
Let
\[ A = \begin{bmatrix} C & 0 \\ B & D \end{bmatrix} \]
be block unitriangular, where $B$ is an $n \times m$ submatrix.  Suppose also
that $C$ and $D$ are orthogonal, and that if $\vv \in \R^n$, then
\[ \vx = A(\vv \oplus \vec0) = C(\vv) \oplus B(\vv) \]
is non-timelike, \ie, $\braket{\vx,\vx}_- \ge 0$.  Then
\[ \norm{A}_\infty \le \phi = \frac{1+\sqrt5}2, \]
where $\norm{A}_\infty$ is the operator norm using the standard
positive-definite inner product
\[ \braket{\vx,\vy}_+ = \sum_j x_j y_j. \]
\label{l:slant} \end{lemma} \eatline

\begin{proof} Since the lemma is invariant under both the left and right
actions of the group $\O(n) \times \O(m)$, we can assume that $C$ and $D$
are the identity, and we can use the singular value decomposition of $B$
to assume that $B$ is diagonal with nonnegative entries.   If $B$ is
diagonal, then we can decompose $A$ as a direct sum
\[ A = A_1 \oplus A_2 \oplus \ldots \oplus A_q \]
of $1 \times 1$ and $2 \times 2$ submatrices that satisfy the hypotheses
of the lemma, where for each summand $A_j$, the signature $(n_j,m_j)$
is either $(1,0)$, $(0,1)$, or $(1,1)$.   Moreover,
\[ \norm{A}_\infty = \max_j \norm{A_j}. \]
We can thus replace $A$ by any given summand $A_j$, and reduce the lemma
to the special cases where $\max(n,m) = 1$, and the sole entry of $B$
(if it exists) is nonnegative.

The cases where $(n,m)$ is $(1,0)$ or $(0,1)$ are trivial, leaving only
the case $(n,m) = (1,1)$.  In this case
\[ A = \begin{bmatrix} 1 & 0 \\ b & 1 \end{bmatrix} \]
with $0 \le b \le 1$.  It is easy to check in this case that
$\norm{A}_\infty$ is maximized when $b=1$ and equals $\phi$.
\end{proof}

\begin{proof}[Proof of \Thm{th:bilip}] Since the map $\mu$ is a
diffeomorphism between the Riemannian manifolds $P \times K$ and $G$,
we can check the bi-Lipschitz condition by establishing uniform bounds
on the $\norm{\cdot}_\infty$ norms of the differentials $d\mu(p,k)$ and
$d\mu^{-1}(g)$ at every point $g = \sigma(p)k$.

Given $g = \sigma(p)k$, we can isometrically identify
\begin{align*}
T_p(P) &\cong T_1(P) \cong L_P \\
T_k(K) &\cong T_1(K) \cong L_K \\
T_g(G) &\cong T_1(G) \cong L_P \oplus L_K
\end{align*}
using (respectively) the action of $p^{-1} \in G$ on $P$, left or right
multiplication of $k^{-1}$ on $K$, and left multiplication of $g^{-1}$
on $G$.  With these identifications, we can interpret $d\mu(p,k)$ as a
linear map
\[ d\mu(p,k):L_P \oplus L_K \to L_P \oplus L_K. \]
By the construction of $\mu$, the linear map $d\mu(p,k)$ takes $L_K$ to
$L_K$ isometrically. $d\mu(p,k)$ does not map $L_P$ to $L_P$ in general.
However, if we regard the differential $d\pi_P(g)$ of the projection
$\pi_P:G \to P$ as a linear map
\[ d\pi_P(g):L_P \oplus L_K \to L_P, \]
it has kernel $L_K$ and the composition $d\pi_P(g) \circ d\mu(p,k)$ is an
isometry from $L_P$ to $L_P$.  Finally $d\mu(p,k)(L_P)$ is non-timelike
in $T_g(G)$ with respect to $\braket{\cdot,\cdot}_{G,-}$ by hypothesis.

Thus, if we choose orthonormal bases for $L_P$ and $L_K$, the linear map
$A = d\mu(p,k)$ becomes a matrix that satisfies all of the hypotheses
of \Lem{l:slant}.  Moreover, if
\[ A = \begin{bmatrix} C & 0 \\ B & D \end{bmatrix} \]
then
\[ A^{-1} = \begin{bmatrix} C^{-1} & 0 \\
    -D^{-1}BC^{-1} & D^{-1} \end{bmatrix} \]
does too.  Thus by the lemma, both $\mu$ and $\mu^{-1}$ are $\phi$-Lipschitz.
\end{proof}

\section{Appendix: Elkasapy's theorem}
\label{s:elkasapy}

In this section, we state and prove Elkasapy's theorem from his article
\cite{Elkasapy:lower}, which was never published in a journal.

\begin{theorem}[Elkasapy {\cite[Lems.~2.2~\&~2.5]{Elkasapy:lower}}]
There is a sequence of words $\omega_1,\omega_2,\ldots \in F_2$ such that
$\nil(\omega_n) = f_n$ is the $n$th Fibonacci number, and such that
\[ \len(\omega_n) = \begin{cases}
    (13 \cdot 2^{n-2} + 2)/7 & n \equiv 0 \pmod 3 \\
    (13 \cdot 2^{n-2} + 4)/7 & n \equiv 1 \pmod 3 \\
    (13 \cdot 2^{n-2} - 6)/7 & n \equiv 2 \pmod 3
    \end{cases} \]
when $n \ge 2$.  Thus, $\lambda \le \log_\phi(2)$.
\label{th:elkasapy} \end{theorem}

In the statement of \Thm{th:elkasapy}, we use the standard indexing of the
Fibonacci numbers with $f_1 = f_2 = 1$.  It seems possible that the theorem
yields the exact minimum value of $\len(\omega)$ whenever $\nil(\omega)
= f_n$ is a Fibonacci number.  Less sharply, Elkasapy conjectures that
$\lambda = \log_\phi(2)$.  (Note that $\log_\phi(2)$ is transcendental
according to the Gelfond--Schneider theorem \cite{Gelfond:hilbert}.)

The proof of \Thm{th:elkasapy} is based on the following recursive
construction of two words $\omega_n,\zeta_n \in F_2$. We set
\begin{eq}{e:elk1} \begin{aligned}
(\omega_1,\zeta_1) &= (g,h^{-1}g^{-1}) \\
(\omega_{n+1},\zeta_{n+1}) &= (\omega_n^{-1} \zeta_n^{-1},\omega_n \zeta_n).
\end{aligned} \end{eq}
Given that
\begin{align*}
\omega_{n+2} &= \omega_{n+1}^{-1} \zeta_{n+1}^{-1}
    = \omega_{n+1}^{-1} \cdot \omega_n \omega_n^{-1}
        \cdot \zeta_n^{-1} \omega_n^{-1} \\
    &= \omega_{n+1}^{-1} \cdot \omega_n
        \cdot \omega_n^{-1} \zeta_n^{-1} \cdot \omega_n^{-1}
    = \omega_{n+1}^{-1} \omega_n \omega_{n+1} \omega_n^{-1},
\end{align*}
we can also define the sequence $(\omega_n)$ with a second recurrence:
\begin{eq}{e:elk2} \omega_1 = g \qquad \omega_2 = h
    \qquad \omega_{n+2} = \comm{\omega_{n+1}^{-1},\omega_n}. \end{eq}
For the moment, we nearly ignore the initial conditions, assuming only
that $\omega_1$ and $\zeta_1$, equivalently $\omega_1$ and $\omega_2$,
do not commute and thus generate a nonabelian free group.  Even with
just this, the recurrence \eqref{e:elk1} tells us that $\len(\omega_n) =
O(2^n)$, while the recurrence \eqref{e:elk2} tells us that
\begin{eq}{e:nilfib} \nil(\omega_{n+2})
    \ge \nil(\omega_{n+1})+\nil(\omega_n), \end{eq}
and thus that $\nil(\omega_n) = \Omega(\phi^n)$.  We thus already learn the
most important estimate, that $\lambda \le \log_\phi(2)$.

To prove the precise formula for $\len(\omega_n)$ in \Thm{th:elkasapy}, we
use the following lemma which gives both $\len(\omega_n)$ and $\len(\zeta_n)$
as well as the beginning and ending letters of both types of words.

\begin{lemma} With $\omega_n$ and $\zeta_n$ as above, if $n = 3k$, then:
\begin{align*}
\omega_n &= h^{-1}\cdots hg^{-1}
    & \len(\omega_n) &= (13 \cdot 2^{n-2} + 2)/7 \\
\zeta_n &= h\cdots h^{-1}g^{-1}
    & \len(\zeta_n) &= (13 \cdot 2^{n-2} + 2)/7.
\end{align*}
If $n = 3k+1 \ge 4$, then
\begin{align*}
\omega_n &= gh^{-1}\cdots h^{-1}
    & \len(\omega_n) &= (13 \cdot 2^{n-2} + 4)/7 \\
\zeta_n &= h^{-1}\cdots h^{-1}g^{-1}
    & \len(\zeta_n) &= (13 \cdot 2^{n-2} + 4)/7.
\end{align*}
If $n = 3k+2$, then
\begin{align*}
\omega_n &= h \cdots h & \len(\omega_n)
    &= (13 \cdot 2^{n-2} - 6)/7 \\
\zeta_n &= gh^{-1}\cdots h^{-1}g^{-1}
    & \len(\zeta_n) &= (13 \cdot 2^{n-2} + 8)/7.
\end{align*}
\label{l:endslen} \end{lemma} \eatline
\begin{proof} It is routine to prove \Lem{l:endslen} by induction on $n$,
using the recurrence \eqref{e:elk1}.  We can directly verify the case $n=2$
given that
\[ (\omega_2,\zeta_2) = (h,gh^{-1}g^{-1}). \]
After that, each case implies the next one, going in a circle mod 3.
\end{proof}

Finally, we establish the precise value of $\nil(\omega_n)$ in
\Thm{th:elkasapy}, with the following lemma.

\begin{lemma} With $\omega_n$ and $f_n$ as above,
\[ \nil(\omega_n) = \ccan_{\SU(2)}(\omega_n) = f_n. \]
\label{l:nilfib} \end{lemma} \eatline

(Note that Elkasapy showed that $\nil(\omega_n) \ge f_n$ and conjectured
that it is an equality.)

\begin{proof} We will show that
\[ f_n \le \nil(\omega_n) \le \ccan_{\SU(2)}(\omega_n) \le f_n. \]
Since
\[ (\omega_1,\omega_2) = (g,h) \qquad \nil(\omega_1) = \nil(\omega_2) = 1, \]
the inequality $f_n \le \nil(\omega_n)$ follows by induction from
\eqref{e:nilfib}.  Meanwhile, the inequality $\nil(\omega_n) \le
\ccan_{\SU(2)}(\omega_n)$ is part of \Thm{th:cannil}.

Finally, we will show that $\ccan_{\SU(2)}(\omega_n) \le f_n$ by evaluating
$\omega_n$ at particular conjugate group elements $g,h \in \SU(2)$ which
we make from Pauli matrices.  Let
\begin{eq}{e:ghval} \begin{aligned}
\omega_1(g,h) = g &= \exp(\eps i Z/2) \\
\omega_2(g,h) = h &= \exp(\eps i Y/2),
\end{aligned} \end{eq}
where $\eps > 0$ is small and can depend on $n$. (Or we can work in
$\SU(2)[[\eps]]$.)  We substitute the matrix values \eqref{e:ghval}
into the recurrence \eqref{e:elk2}, and then simplify the result with
the approximation \eqref{e:glbracket} and the commutation relations for
Pauli matrices:
\[ [X,Y] = 2iZ \qquad
[Y,Z] = 2iX \qquad
[Z,X] = 2iY. \]
We learn by induction that
\[ \omega_n(g,h) = \begin{cases}
\displaystyle \exp\bigl(\frac{\eps^{f_n} i Z}2
    + O(\eps^{f_n+1})\bigr) & n \equiv 1 \pmod 3  \\[1.5ex]
\displaystyle \exp\bigl(\frac{\eps^{f_n} i Y}2
    + O(\eps^{f_n+1})\bigr) & n \equiv 2 \pmod 3 \\[1.5ex]
\displaystyle \exp\bigl(\frac{\eps^{f_n} i X}2
    + O(\eps^{f_n+1})\bigr) & n \equiv 0 \pmod 3 \end{cases} \]
for all $n$.  In particular,
\[ \omega_n(g,h) \ne \exp(O(\eps^{f_n+1})), \]
which shows that $\ccan_{\SU(2)}(\omega_n) \le f_n$.
\end{proof}

\vspace{1in} \mbox{ }

\bibliography{books,gr,gt,na,nt,rt,qp,me}

\end{document}